\documentclass[a4paper,UKenglish,cleveref, autoref, thm-restate]{lipics-v2021}

\pdfoutput=1 
\hideLIPIcs  
\nolinenumbers

\usepackage{amsmath,amssymb,amsfonts}
\usepackage{algorithm}

\usepackage{graphicx}
\usepackage{textcomp}
\usepackage{xcolor}
\usepackage{pgf, tikz}
\usetikzlibrary{arrows, automata, petri}
\usepackage{longtable}
\usepackage{todonotes}
\usepackage{subcaption}
\usepackage{float}

\usepackage{amsthm}
\usepackage[noend]{algpseudocode}

\newcommand{\mybraces}[1]{\left(#1\right)}
\newcommand{\myceil}[1]{\left\lceil #1\right\rceil}

\newcommand{\myMax}[1]{\max\{#1\}}
\usepackage{xpatch}
\usepackage{thmtools}
\usepackage{apptools}

\makeatletter
\xpatchcmd{\thmt@restatable}
{\csname #2\@xa\endcsname\ifx\@nx#1\@nx\else[{#1}]\fi}
{\IfAppendix{\csname #2\@xa\endcsname}{\csname #2\@xa\endcsname\ifx\@nx#1\@nx\else[{#1}]\fi}}
{}{} 
\makeatother

\newcommand{\io}{I/O}

\newcommand{\myIe}{\emph{I.e.,}}

\newcommand{\alg}{\mathcal{A}}

\newcommand{\dom}{D}

\date{}
\title{On the I/O Complexity of the CYK Algorithm and of a Family of Related DP Algorithms} 

\author{Lorenzo De Stefani}{Brown University, Department of Computer Science, USA}{lorenzo\_destefani@brown.edu}{}{}
\authorrunning{Lorenzo De Stefani and Vedant Gupta}

\author{Vedant Gupta}{Brown University, Department of Computer Science, USA}{joanrpublic@dummycollege.org}{}{}

\Copyright{Jane Open Access and Joan R. Public} 



\begin{document}

\maketitle

\begin{abstract}
Asymptotically tight lower bounds are derived for the Input/Output (\io{}) complexity of a class of dynamic programming algorithms including matrix chain multiplication, optimal polygon triangulation, and the construction of optimal binary search trees. Assuming no recomputation of intermediate values, we establish an $\Omega\mybraces{\frac{n^3}{\sqrt{M}B}}$ \io{} lower bound, where $n$ denotes the size of the input and $M$ denotes the size of the available fast memory (cache). When recomputation is allowed, we show the same bound holds for $M < cn$, where $c$ is a positive constant. In the case where $M \ge 2n$, we show an $\Omega\mybraces{n/B}$ \io{} lower bound. We also discuss algorithms for which the number of executed I/O operations matches asymptotically each of the presented lower bounds, which are thus asymptotically tight.

Additionally, we refine our general method to obtain a lower bound for the \io{} complexity of the Cocke-Younger-Kasami algorithm, where the size of the grammar impacts the \io{} complexity. An upper bound with asymptotically matching performance in many cases is also provided.
\end{abstract}

\section{Introduction}\label{sec:introduction}
The performance of computing systems is significantly influenced by data movement, affecting both time and energy consumption. This ongoing technological trend~\cite{patterson2005getting} is expected to continue as fundamental limits on minimum device size and maximum message speed lead to inherent costs associated with data transfer, be it across the levels of a hierarchical memory system or between processing elements in a parallel system~\cite{bilardi1995horizons}. While the communication requirements of algorithms have been extensively studied, deriving significant lower bounds based on these requirements and matching upper bounds remains a critical challenge.

Dynamic Programming (DP) is a classic algorithmic framework introduced by Bellman~\cite{bellman1954theory} that is particularly effective in optimization tasks where problems can be broken down into overlapping subproblems with optimal substructure and storing the solutions to these subproblems to prevent redundant work. DP is widely used in computer science and other important fields such as control theory, operations research, and computational biology.

In this work, we analyze the memory access cost (\textit{\io{} complexity}) of a family of DP  algorithms when executed in a two-level storage hierarchy with $M$ words of fast memory (cache). The analyzed algorithms share a similar structure of dependence on the results of subproblems and include, among others, the classic DP algorithms for Matrix Chain Multiplication, Boolean Parethesization, Optimal Polygon Triangulation, and Construction of Optimal Binary Search Trees.

We obtain asymptotically tight lower bounds for the \io{} complexity of these algorithms both for schedules in which no intermediate values are computed more than once and for general schedules allowing recomputation.  Our analysis reveals that for a range of values of the size of the cache, for the algorithms under consideration, recomputing intermediate values can lead to an asymptotically lower number of \io{} operations. 

 We further refine our analysis to derive an \io{} lower bound for the  Cocke-Younger-Kasami (CYK) algorithm \cite{cocke1969programming, younger1967recognition, kasami1966efficient} and an (almost) asymptotically matching upper bound.



\noindent\textbf{Previous and related work:} 
\io{} complexity was introduced in the seminal work by Hong and
Kung~\cite{jia1981complexity}; it denotes the number of data
transfers between the two levels of a memory hierarchy with a cache of $M$ words and a slow memory of unbounded size. Hong and Kung presented techniques to develop lower bounds for the \io{} complexity of computations modeled by \emph{Computational Directed Acyclic Graphs} (CDAGs).  

The ``\emph{edge expansion technique}'' of~\cite{ballard2012graph}, the ``\emph{path
routing technique}'' of~\cite{scott2015matrix}, and the ``\emph{closed dichotomy
width}'' technique of~\cite{bilardi1999processor} all yield \io{} lower
bounds that apply only to computational schedules for which no
intermediate result is ever computed more than once
(\emph{nr-computations}) it is also
important to investigate what can be achieved with recomputation. In
fact, for some CDAGs, recomputing intermediate values reduces the
space and/or the \io{} complexity of an
algorithm~\cite{savage1995extending}.
In~\cite{bilardi2001characterization}, it is shown that some
algorithms admit a \emph{portable schedule} (i.e., a schedule which
achieves optimal performance across memory hierarchies with different
access costs) only if recomputation is allowed. Recomputation can also enhance the performance of simulations among
networks (see~\cite{KochLMRRS97} and references therein) and
plays a key role in the design of efficient area-universal VLSI architectures.

A number of lower-bound techniques
that allow for recomputation have been presented in the literature,
including the ``\emph{S-partition}
technique''~\cite{jia1981complexity}, the ``\emph{S-span}
technique''~\cite{savage1995extending}, 
technique'' , the ``\emph{S-covering}
technique''~\cite{bilardi2000space} (which merges and extends aspects
from both~\cite{jia1981complexity} and~\cite{savage1995extending}), the ``\emph{G-flow} technique''~\cite{bilardi2017complexity,bdstksoda19}, and the ``\emph{visit partition technique}~\cite{bilardi_et_al:LIPIcs.FSTTCS.2022.7}. However, to the best of our knowledge, none of these have been previously applied to the family of DP algorithms considered in this work.

A systematic study of \textit{Dynamic Programming} (DP) was introduced by Bellman \cite{bellman1954theory}, and has since become a popular problem-solving methodology. Galil and Park \cite{galil1992dynamic} produce a survey and classifications of popular dynamic programming algorithms and optimization methods. 
Several results have been presented in the literature aimed at optimizing the running time of Dynamic Programming algorithms \cite{hu1981computation, knuthbst, yao1980efficient, galil1992dynamic}. Notably, Valiant \cite{valiant1974general} showed that context-free recognition can be carried out as fast as boolean matrix multiplication.

Cherng and Lander \cite{cherng2005cache} provide a cache-oblivious divide-and-conquer algorithm and a cache-aware blocked algorithm for \textit{simple dynamic programs}, which includes matrix chain multiplication and the construction of optimal binary search trees. Their approach is based on Valiant's algorithm \cite{valiant1974general} and incurs $\mathcal{O}(\frac{n^3}{B\sqrt{M}})$ cache misses.  In \cite{chowdhury2006cache, chowdhury2007cache}, Choudhury and Ramachandran introduce the \textit{Gaussian Elimination Paradigm (GEP)} that provides cache-oblivious algorithms for problems including All-Pairs Shortest Path. The authors also provide a way to transform simple DP problems into GEP problems. A proof of the optimality of the proposed algorithm for GEP problems is provided based on a similar result given in~\cite{jia1981complexity} for the matrix multiplication problem. As we show, this result does not extend to simple DP, as fewer cache misses can be achieved in some scenarios. 

Cache-oblivious algorithms for other DP problems have been explored \cite{park2004optimizing, chowdhury2008cache, demaine2017fine}. There has also been work on parallelizing cache-efficient DP algorithms \cite{chowdhury2008cachemulti, chowdhury2017provably, javanmard2019toward, tang2015cache, tang2011easypdp, farivar2012algorithm, ding2024parallel, blleloch2018improved}, experimental evaluations of the resulting performance improvement \cite{ramachandran2007cache}, and also the automated discovery of cache-efficient DP algorithms \cite{chowdhury2017autogen, itzhaky2016deriving}.
 
Dunlop et. al \cite{dunlop2010reducing} provides an empirical analysis of various strategies to reduce the number of accesses to the grammar for the \textit{CYK} algorithm. There have also been other works on modifying a grammar representation to improve efficiency. Notably, Song et. al \cite{song2008better} provide an empirical analysis of different grammar binarization approaches to improve efficiency, while Lange et. al \cite{lange2009cnf} argue for the merits of a modified version of CYK to work with grammars in \textit{binary normal form} instead of \textit{Chomsky normal form}.\\

\noindent\textbf{Summary of results:} 
In this work, we analyze the memory access cost (\io{} complexity) of a family of DP  algorithms when executed in a two-level storage hierarchy with $M$ words of cache. The analyzed algorithms share a similar structure of dependence on the results of subproblems while allowing for 
diverse implementation choices and optimal problem substructure. 

We analyze both schedules in which no intermediate values are computed more than once (no-recomputation schedules) and general ones that allow for recomputation. 
For both settings, we derive an $\Omega\mybraces{\frac{n^3}{\sqrt{M}B}}$  lower bound to the number of \io{} operations required by sequential execution of these algorithms on an input of size $n$ in a two-level storage hierarchy with a cache of size $M\leq cn$, for some positive constant value $c$ and such that $B\geq 1$ memory words can be moved between consecutive memory locations of the cache and of the slow memory with a single \io{} operation. However, we show that while no-recomputation schedules require $\Omega\mybraces{\frac{n^3}{\sqrt{M}B}}$ 
\io{} operations even when using a cache of size at least $2n$ and $o\mybraces{n^2}$, schedules using recomputation achieve an asymptotic reduction of the number of required \io{} operations by a factor $\Theta\left(n^2/\sqrt{M}\right)$. This is particularly significant as in many cases of interest (e.g. Matrix Multiplication~\cite{jia1981complexity,ballard2010communication,bilardi2017complexity}, Fast Fourier Transform~\cite{jia1981complexity}, Integer Multiplication~\cite{bdstksoda19}) recomputation has shown to enable a reduction of the \io{} cost by at most a constant multiplicative factor.

 These results are obtained by analyzing a representation of the considered algorithms as Computational Directed Acyclic Graphs (CDAGs) whose structure captures the structure of dependence between subproblems which is common to the DP algorithms considered in our work. We provide a general construction of such CDAGs and analyze their internal connection properties.

We refine our general method to obtain an $\Omega\mybraces{\frac{n^3\Gamma}{B\sqrt{M}}}$ lower bound for the \io{} complexity of the Cocke-Younger-Kasami algorithm when deciding whether an input string of length  $n$ is a member of the language generated by a given Context-Free Grammar with $\Gamma$ rules with distinct right-hand-sides not including terminals.  While CYK exhibits a structure of subproblem dependencies similar to that of previously analyzed problems, it presents several challenging differences which we address by modifying our lower-bound methods by further refining its CDAG representation. 
Finally, we present a cache-oblivious implementation of the Cocke-Younger-Kasami algorithm for which, depending on the composition of the rules of the considered CFG, the number of \io{} operations to be executed \emph{almost} asymptotically matches the lower bound.

\section{Preliminaries}\label{sec:preliminaries}
We consider a family of Dynamic Programming (DP) algorithms following the general strategy outlined in \emph{Prototype Algorithm}~\ref{alg:proto}: Given an input of size $n$, the solution, denoted as $S\mybraces{1,n}$ is computed bottom up. First, the values corresponding to the solution of subproblems of input size one, $S\mybraces{i,i}$ for $i\in\{1,2,\ldots,n\}$, are initialized to some set default value. The results of subproblems corresponding to parts of the input of growing size $\ell$, that is $S\mybraces{i,i+\ell}$ for $i\in\{1,2,\ldots,n\}$ and $\ell\in\{1,2,\ldots,n-i\}$, are computed by first evaluating the $\ell$ compositions of pairs of subproblems $S\mybraces{i,i+k}$ and $S\mybraces{i+k+1, i+\ell}$, for $k\in\{0,1,\ldots \ell-1\}$ and then composing them according to the optimal subproblem structure used in the algorithm. 


\begin{algorithm}[t]
\caption{Prototype DP algorithm $\mathcal{A}^*$}\label{alg:proto}
\footnotesize
\begin{algorithmic}[1]
\State \textbf{Input} $\{x_0, x_1, ..., x_n\}$
\State \textbf{Output} $S\mybraces{1,n}$
\For{$i=1$ \textbf{to} $n$} \Comment{Initialization of subproblems $S\mybraces{i,i}$}
\State $S(i,i)\leftarrow\textsc{initialization\_value}$
\EndFor

\For{$l=2$ \textbf{to} $n$}
    \For{$i=1$  \textbf{to} $n-l+1$}
        \State $j\leftarrow i + l - 1$
        \State $S\mybraces{i,j}\leftarrow \textsc{least\_optimal\_value}$\Comment{Initialization of accumulator subproblem $S\mybraces{i,j}$} 
        \For{$k=i$  \textbf{to} $j - 1$}
            \State $q\leftarrow \textsc{COMBINE}\mybraces{S\mybraces{i, k},S\mybraces{k+1,j}}$
                \State $S\mybraces{i,j} \leftarrow \textsc{AGGREGATE}\mybraces{S\mybraces{i,j},q}$
        \EndFor
    \EndFor
\EndFor
\State \textbf{return} $S\mybraces{1,n}$
\end{algorithmic}
\end{algorithm}

Depending on the specific problems, the default initialization values (\texttt{INITIALIZATION\_VALUE} in the Prototype  Algorithm) of $S\mybraces{i,i}$, the way that the pairs $S\mybraces{i,i+k}$ and $S\mybraces{i+k+1, i+\ell}$ are combined (\texttt{COMBINE}), and the way these results are composed to evaluate $S\mybraces{i,i+\ell}$ (\texttt{AGGREGATE}) may differ. However, \texttt{AGGREGATE} is assumed to be commutative and associative. Besides these differences,, all the considered subproblems share the subproblem dependence structure outlined in the Prototype Algorithm for which a visual representation is provided in Figure~\ref{fig:simpledag}. Examples of algorithms following such structure include the classic DP algorithm for the Matrix Chain Multiplication problem (Algorithm~\ref{alg:MCM}), the Optimal Convex Polygon Triangulation (Algorithm~\ref{alg:OPT}), the construction of Optimal Binary Search Trees (Algorithm~\ref{alg:BST}), and the Cocke-Younger-Kasami algorithm (Algorithm~\ref{alg:CYK}). In our main contribution, we provide a general analysis of the \io{} complexity of algorithms following the structure of the Prototype Algorithm $\mathcal{A}^*$, and we then refine our analysis for the mentioned algorithms of interest.\\

\noindent\textbf{Computational Directed Acyclic Graphs:}
Our analysis is based on modeling the execution of the DP algorithms of interest 
as a \emph{Computational Directed Acyclic Graph} (CDAG) $G=\left(V,E\right)$.
Each vertex $v\in V$ represents either an input value or the
result of an operation (i.e., an intermediate result or one
of the output values) which is stored using a single memory word. 
The \emph{directed} edges in $E$ represent data dependencies. That is, a pair of vertices $u,v\in V$ are connected by an edge $(u,v)$ directed from $u$ to $v$ if and only if the value corresponding to $u$ is an operand of the unit time operation which computes the value corresponding to $v$. $u$ is said to be a \emph{predecessor} of $v$ and $v$ is said to be a \emph{successor} of $u$. For a given vertex $v$, the \emph{set of its predecessors} (resp., \emph{successors}) is defined as $pre(v)=\{u\in V\ s.t. (u,v)\in E\}$ (resp., $suc(v)=\{u\in V\ s.t. (v,u)\in E\}$. We refer to vertices with no predecessors (resp., successors) as the \emph{input} (resp., \emph{output}) vertices of a CDAG.


We say that $G'=\left(V',E'\right)$ is a \emph{sub-CDAG} of $G=\left(V,E\right)$ if
$V'\subseteq V$ and $E' \subseteq E \cap (V'\times V')$. We say that two sub-CDAGs  $G'=\left(V',E'\right)$ and $G''=\left(V'',E''\right)$ of $G$ are \emph{vertex-disjoint} if $V'\cap V''=\emptyset$.\\



\noindent\textbf{\io{} model:}\label{sec:iomodel}
We assume that sequential computations are executed on a system with a two-level memory hierarchy, consisting of a cache of size $M$ (measured in memory words) and a \emph{slow memory} of unlimited size. An operation can be executed only if all its operands are in the cache. 
We assume that the processor is equipped with standard elementary logic and algebraic operations. Data can be moved from the slow memory to the cache by \texttt{read} operations, and, in the other direction, by \texttt{write} operations. These operations are called \emph{\io{} operations}. We assume that the input data is stored in the slow memory at the beginning of the computation. The evaluation of a CDAG in this model can be analyzed using the ``\emph{red-blue pebble game}''~\cite{jia1981complexity}. The number of \io{} operations executed when evaluating a CDAG depends on the ``\emph{computational schedule}'', that is, on the order in which vertices are evaluated and on which values are kept in/discarded from the cache. 

The \emph{\io{} complexity} of a CDAG $G$ corresponding to the execution of algorithm $\mathcal{A}$ for an input of size $n$, denoted as $IO_{\mathcal{A}}(M)$, is defined as the minimum number of \io{} operations over all possible computational schedules.
We further consider a generalization of this model known as the ``\emph{External Memory Model}'' by Aggarwal and Vitter~\cite{Aggarwal:1988:ICS:48529.48535}
where $B\geq 1$ memory words can be moved between consecutive memory locations of the cache and of the
slow memory with a single \io{} operation. 


\section{CDAG constuction}\label{sec:DAG}

\begin{figure}
    \centering
    \begin{subfigure}[t]{0.46\textwidth}
        \centering
        \resizebox{\textwidth}{!}{%
            \tikzset{every picture/.style={line width=0.75pt}} 

\begin{tikzpicture}[x=0.75pt,y=0.75pt,yscale=-1,xscale=1]

\draw  [fill={rgb, 255:red, 0; green, 0; blue, 0 }  ,fill opacity=1 ] (77,158) .. controls (77,152.48) and (81.48,148) .. (87,148) .. controls (92.52,148) and (97,152.48) .. (97,158) .. controls (97,163.52) and (92.52,168) .. (87,168) .. controls (81.48,168) and (77,163.52) .. (77,158) -- cycle ;
\draw  [fill={rgb, 255:red, 0; green, 0; blue, 0 }  ,fill opacity=1 ] (137,98) .. controls (137,92.48) and (141.48,88) .. (147,88) .. controls (152.52,88) and (157,92.48) .. (157,98) .. controls (157,103.52) and (152.52,108) .. (147,108) .. controls (141.48,108) and (137,103.52) .. (137,98) -- cycle ;
\draw  [fill={rgb, 255:red, 0; green, 0; blue, 0 }  ,fill opacity=1 ] (197,158) .. controls (197,152.48) and (201.48,148) .. (207,148) .. controls (212.52,148) and (217,152.48) .. (217,158) .. controls (217,163.52) and (212.52,168) .. (207,168) .. controls (201.48,168) and (197,163.52) .. (197,158) -- cycle ;
\draw  [fill={rgb, 255:red, 0; green, 0; blue, 0 }  ,fill opacity=1 ] (317,158) .. controls (317,152.48) and (321.48,148) .. (327,148) .. controls (332.52,148) and (337,152.48) .. (337,158) .. controls (337,163.52) and (332.52,168) .. (327,168) .. controls (321.48,168) and (317,163.52) .. (317,158) -- cycle ;
\draw  [fill={rgb, 255:red, 0; green, 0; blue, 0 }  ,fill opacity=1 ] (257,98) .. controls (257,92.48) and (261.48,88) .. (267,88) .. controls (272.52,88) and (277,92.48) .. (277,98) .. controls (277,103.52) and (272.52,108) .. (267,108) .. controls (261.48,108) and (257,103.52) .. (257,98) -- cycle ;
\draw  [fill={rgb, 255:red, 0; green, 0; blue, 0 }  ,fill opacity=1 ] (197,38) .. controls (197,32.48) and (201.48,28) .. (207,28) .. controls (212.52,28) and (217,32.48) .. (217,38) .. controls (217,43.52) and (212.52,48) .. (207,48) .. controls (201.48,48) and (197,43.52) .. (197,38) -- cycle ;
\draw  [fill={rgb, 255:red, 0; green, 0; blue, 0 }  ,fill opacity=1 ] (17,218) .. controls (17,212.48) and (21.48,208) .. (27,208) .. controls (32.52,208) and (37,212.48) .. (37,218) .. controls (37,223.52) and (32.52,228) .. (27,228) .. controls (21.48,228) and (17,223.52) .. (17,218) -- cycle ;
\draw  [fill={rgb, 255:red, 0; green, 0; blue, 0 }  ,fill opacity=1 ] (257,218) .. controls (257,212.48) and (261.48,208) .. (267,208) .. controls (272.52,208) and (277,212.48) .. (277,218) .. controls (277,223.52) and (272.52,228) .. (267,228) .. controls (261.48,228) and (257,223.52) .. (257,218) -- cycle ;
\draw  [fill={rgb, 255:red, 0; green, 0; blue, 0 }  ,fill opacity=1 ] (137,218) .. controls (137,212.48) and (141.48,208) .. (147,208) .. controls (152.52,208) and (157,212.48) .. (157,218) .. controls (157,223.52) and (152.52,228) .. (147,228) .. controls (141.48,228) and (137,223.52) .. (137,218) -- cycle ;
\draw  [fill={rgb, 255:red, 0; green, 0; blue, 0 }  ,fill opacity=1 ] (377,218) .. controls (377,212.48) and (381.48,208) .. (387,208) .. controls (392.52,208) and (397,212.48) .. (397,218) .. controls (397,223.52) and (392.52,228) .. (387,228) .. controls (381.48,228) and (377,223.52) .. (377,218) -- cycle ;
\draw    (27,218) -- (78.88,166.12) ;
\draw [shift={(81,164)}, rotate = 135] [fill={rgb, 255:red, 0; green, 0; blue, 0 }  ][line width=0.08]  [draw opacity=0] (8.93,-4.29) -- (0,0) -- (8.93,4.29) -- cycle    ;
\draw    (147,98) -- (198.88,46.12) ;
\draw [shift={(201,44)}, rotate = 135] [fill={rgb, 255:red, 0; green, 0; blue, 0 }  ][line width=0.08]  [draw opacity=0] (8.93,-4.29) -- (0,0) -- (8.93,4.29) -- cycle    ;
\draw    (91,154) -- (138.88,106.12) ;
\draw [shift={(141,104)}, rotate = 135] [fill={rgb, 255:red, 0; green, 0; blue, 0 }  ][line width=0.08]  [draw opacity=0] (8.93,-4.29) -- (0,0) -- (8.93,4.29) -- cycle    ;
\draw    (207,158) -- (258.88,106.12) ;
\draw [shift={(261,104)}, rotate = 135] [fill={rgb, 255:red, 0; green, 0; blue, 0 }  ][line width=0.08]  [draw opacity=0] (8.93,-4.29) -- (0,0) -- (8.93,4.29) -- cycle    ;
\draw    (147,218) -- (198.88,166.12) ;
\draw [shift={(201,164)}, rotate = 135] [fill={rgb, 255:red, 0; green, 0; blue, 0 }  ][line width=0.08]  [draw opacity=0] (8.93,-4.29) -- (0,0) -- (8.93,4.29) -- cycle    ;
\draw    (147,218) -- (95.12,166.12) ;
\draw [shift={(93,164)}, rotate = 45] [fill={rgb, 255:red, 0; green, 0; blue, 0 }  ][line width=0.08]  [draw opacity=0] (8.93,-4.29) -- (0,0) -- (8.93,4.29) -- cycle    ;
\draw    (267,218) -- (318.88,166.12) ;
\draw [shift={(321,164)}, rotate = 135] [fill={rgb, 255:red, 0; green, 0; blue, 0 }  ][line width=0.08]  [draw opacity=0] (8.93,-4.29) -- (0,0) -- (8.93,4.29) -- cycle    ;
\draw    (207,158) -- (155.12,106.12) ;
\draw [shift={(153,104)}, rotate = 45] [fill={rgb, 255:red, 0; green, 0; blue, 0 }  ][line width=0.08]  [draw opacity=0] (8.93,-4.29) -- (0,0) -- (8.93,4.29) -- cycle    ;
\draw    (267,218) -- (215.12,166.12) ;
\draw [shift={(213,164)}, rotate = 45] [fill={rgb, 255:red, 0; green, 0; blue, 0 }  ][line width=0.08]  [draw opacity=0] (8.93,-4.29) -- (0,0) -- (8.93,4.29) -- cycle    ;
\draw    (327,158) -- (275.12,106.12) ;
\draw [shift={(273,104)}, rotate = 45] [fill={rgb, 255:red, 0; green, 0; blue, 0 }  ][line width=0.08]  [draw opacity=0] (8.93,-4.29) -- (0,0) -- (8.93,4.29) -- cycle    ;
\draw    (267,98) -- (215.12,46.12) ;
\draw [shift={(213,44)}, rotate = 45] [fill={rgb, 255:red, 0; green, 0; blue, 0 }  ][line width=0.08]  [draw opacity=0] (8.93,-4.29) -- (0,0) -- (8.93,4.29) -- cycle    ;
\draw    (387,218) -- (335.12,166.12) ;
\draw [shift={(333,164)}, rotate = 45] [fill={rgb, 255:red, 0; green, 0; blue, 0 }  ][line width=0.08]  [draw opacity=0] (8.93,-4.29) -- (0,0) -- (8.93,4.29) -- cycle    ;
\draw    (27,218) .. controls (52.61,165.3) and (91.8,122.79) .. (138.85,104.8) ;
\draw [shift={(141,104)}, rotate = 159.97] [fill={rgb, 255:red, 0; green, 0; blue, 0 }  ][line width=0.08]  [draw opacity=0] (8.93,-4.29) -- (0,0) -- (8.93,4.29) -- cycle    ;
\draw    (387,218) .. controls (354.5,157.42) and (332.66,131.29) .. (275.63,105.19) ;
\draw [shift={(273,104)}, rotate = 24.19] [fill={rgb, 255:red, 0; green, 0; blue, 0 }  ][line width=0.08]  [draw opacity=0] (8.93,-4.29) -- (0,0) -- (8.93,4.29) -- cycle    ;
\draw    (327,158) .. controls (294.5,97.42) and (272.66,71.29) .. (215.63,45.19) ;
\draw [shift={(213,44)}, rotate = 24.19] [fill={rgb, 255:red, 0; green, 0; blue, 0 }  ][line width=0.08]  [draw opacity=0] (8.93,-4.29) -- (0,0) -- (8.93,4.29) -- cycle    ;
\draw    (267,218) .. controls (234.49,157.42) and (212.66,131.29) .. (155.63,105.19) ;
\draw [shift={(153,104)}, rotate = 24.19] [fill={rgb, 255:red, 0; green, 0; blue, 0 }  ][line width=0.08]  [draw opacity=0] (8.93,-4.29) -- (0,0) -- (8.93,4.29) -- cycle    ;
\draw    (87,158) .. controls (112.61,105.3) and (151.8,62.79) .. (198.85,44.8) ;
\draw [shift={(201,44)}, rotate = 159.97] [fill={rgb, 255:red, 0; green, 0; blue, 0 }  ][line width=0.08]  [draw opacity=0] (8.93,-4.29) -- (0,0) -- (8.93,4.29) -- cycle    ;
\draw    (147,218) .. controls (172.61,165.3) and (211.8,122.79) .. (258.85,104.8) ;
\draw [shift={(261,104)}, rotate = 159.97] [fill={rgb, 255:red, 0; green, 0; blue, 0 }  ][line width=0.08]  [draw opacity=0] (8.93,-4.29) -- (0,0) -- (8.93,4.29) -- cycle    ;
\draw    (27,218) .. controls (44.82,155.13) and (124.39,58.46) .. (198.75,44.4) ;
\draw [shift={(201,44)}, rotate = 170.54] [fill={rgb, 255:red, 0; green, 0; blue, 0 }  ][line width=0.08]  [draw opacity=0] (8.93,-4.29) -- (0,0) -- (8.93,4.29) -- cycle    ;
\draw    (387,218) .. controls (356.31,135.33) and (289.36,64.92) .. (215.25,44.6) ;
\draw [shift={(213,44)}, rotate = 14.57] [fill={rgb, 255:red, 0; green, 0; blue, 0 }  ][line width=0.08]  [draw opacity=0] (8.93,-4.29) -- (0,0) -- (8.93,4.29) -- cycle    ;

\draw (181,2) node [anchor=north west][inner sep=0.75pt]  [font=\small] [align=left] {

$S( 1,\ 4)$
};
\draw (69,62) node [anchor=north west][inner sep=0.75pt]  [font=\small] [align=left] {
$S( 1,\ 3)$
};
\draw (298,62) node [anchor=north west][inner sep=0.75pt]  [font=\small] [align=left] {
$S( 2,\ 4)$
};
\draw (10,122) node [anchor=north west][inner sep=0.75pt]  [font=\small] [align=left] {
$S( 1,\ 2)$
};
\draw (182,98) node [anchor=north west][inner sep=0.75pt]  [font=\small] [align=left] {
$S( 2,\ 3)$
};
\draw (357,122) node [anchor=north west][inner sep=0.75pt]  [font=\small] [align=left] {
$S( 3,\ 4)$
};
\draw (1,230) node [anchor=north west][inner sep=0.75pt]  [font=\small] [align=left] {
$S( 1,\ 1)$
};
\draw (241,230) node [anchor=north west][inner sep=0.75pt]  [font=\small] [align=left] {
$S( 3,\ 3)$
};
\draw (365,230) node [anchor=north west][inner sep=0.75pt]  [font=\small] [align=left] {
$S( 4,\ 4)$
};
\draw (122,230) node [anchor=north west][inner sep=0.75pt]  [font=\small] [align=left] {
$S( 2,\ 2)$
};

\end{tikzpicture}
        }
        \caption{Representation of the subproblem dependency structure for the Prototype Algorithm and other DP algorithms considered in this work for an input of size $n=4$.}
        \label{fig:simpledag}
    \end{subfigure}
    \hfill
    \begin{subfigure}[t]{0.46\textwidth}
        \centering
        \resizebox{\textwidth}{!}{%
            \tikzset{every picture/.style={line width=0.75pt}} 

\begin{tikzpicture}[x=0.75pt,y=0.75pt,yscale=-1,xscale=1]

\draw   (309.14,31) -- (340,80) -- (280,80) -- cycle ;
\draw [color={rgb, 255:red, 0; green, 0; blue, 0 }  ,draw opacity=1 ]   (130,250) .. controls (154.97,183.14) and (200.31,121.71) .. (269.88,81.22) ;
\draw [shift={(272,80)}, rotate = 150.24] [fill={rgb, 255:red, 0; green, 0; blue, 0 }  ,fill opacity=1 ][line width=0.08]  [draw opacity=0] (8.93,-4.29) -- (0,0) -- (8.93,4.29) -- cycle    ;
\draw [color={rgb, 255:red, 0; green, 0; blue, 0 }  ,draw opacity=1 ]   (370,130) .. controls (312.4,121.2) and (310.29,121.02) .. (281.78,89.94) ;
\draw [shift={(280,88)}, rotate = 47.53] [fill={rgb, 255:red, 0; green, 0; blue, 0 }  ,fill opacity=1 ][line width=0.08]  [draw opacity=0] (8.93,-4.29) -- (0,0) -- (8.93,4.29) -- cycle    ;
\draw [color={rgb, 255:red, 0; green, 0; blue, 0 }  ,draw opacity=1 ]   (490,250) .. controls (467.23,184.08) and (422.14,122.7) .. (350.19,81.25) ;
\draw [shift={(348,80)}, rotate = 29.52] [fill={rgb, 255:red, 0; green, 0; blue, 0 }  ,fill opacity=1 ][line width=0.08]  [draw opacity=0] (8.93,-4.29) -- (0,0) -- (8.93,4.29) -- cycle    ;
\draw [color={rgb, 255:red, 0; green, 0; blue, 0 }  ,draw opacity=1 ]   (249.94,130) .. controls (306.08,119.24) and (310.08,119.99) .. (338.25,89.88) ;
\draw [shift={(340,88)}, rotate = 132.92] [fill={rgb, 255:red, 0; green, 0; blue, 0 }  ,fill opacity=1 ][line width=0.08]  [draw opacity=0] (8.93,-4.29) -- (0,0) -- (8.93,4.29) -- cycle    ;
\draw [color={rgb, 255:red, 0; green, 0; blue, 0 }  ,draw opacity=1 ]   (190,190) .. controls (258.53,163.29) and (259.22,162.04) .. (308.49,90.2) ;
\draw [shift={(310,88)}, rotate = 124.45] [fill={rgb, 255:red, 0; green, 0; blue, 0 }  ,fill opacity=1 ][line width=0.08]  [draw opacity=0] (8.93,-4.29) -- (0,0) -- (8.93,4.29) -- cycle    ;
\draw [color={rgb, 255:red, 0; green, 0; blue, 0 }  ,draw opacity=1 ]   (430,190) .. controls (361.91,164.28) and (360.25,162.06) .. (311.49,90.2) ;
\draw [shift={(310,88)}, rotate = 55.84] [fill={rgb, 255:red, 0; green, 0; blue, 0 }  ,fill opacity=1 ][line width=0.08]  [draw opacity=0] (8.93,-4.29) -- (0,0) -- (8.93,4.29) -- cycle    ;
\draw  [fill={rgb, 255:red, 0; green, 0; blue, 0 }  ,fill opacity=1 ] (180,190) .. controls (180,184.48) and (184.48,180) .. (190,180) .. controls (195.52,180) and (200,184.48) .. (200,190) .. controls (200,195.52) and (195.52,200) .. (190,200) .. controls (184.48,200) and (180,195.52) .. (180,190) -- cycle ;
\draw  [fill={rgb, 255:red, 0; green, 0; blue, 0 }  ,fill opacity=1 ] (240,130) .. controls (240,124.48) and (244.48,120) .. (250,120) .. controls (255.52,120) and (260,124.48) .. (260,130) .. controls (260,135.52) and (255.52,140) .. (250,140) .. controls (244.48,140) and (240,135.52) .. (240,130) -- cycle ;
\draw  [fill={rgb, 255:red, 0; green, 0; blue, 0 }  ,fill opacity=1 ] (300,190) .. controls (300,184.48) and (304.48,180) .. (310,180) .. controls (315.52,180) and (320,184.48) .. (320,190) .. controls (320,195.52) and (315.52,200) .. (310,200) .. controls (304.48,200) and (300,195.52) .. (300,190) -- cycle ;
\draw  [fill={rgb, 255:red, 0; green, 0; blue, 0 }  ,fill opacity=1 ] (420,190) .. controls (420,184.48) and (424.48,180) .. (430,180) .. controls (435.52,180) and (440,184.48) .. (440,190) .. controls (440,195.52) and (435.52,200) .. (430,200) .. controls (424.48,200) and (420,195.52) .. (420,190) -- cycle ;
\draw  [fill={rgb, 255:red, 0; green, 0; blue, 0 }  ,fill opacity=1 ] (360,130) .. controls (360,124.48) and (364.48,120) .. (370,120) .. controls (375.52,120) and (380,124.48) .. (380,130) .. controls (380,135.52) and (375.52,140) .. (370,140) .. controls (364.48,140) and (360,135.52) .. (360,130) -- cycle ;
\draw  [fill={rgb, 255:red, 0; green, 0; blue, 0 }  ,fill opacity=1 ] (300,30) .. controls (300,24.48) and (304.48,20) .. (310,20) .. controls (315.52,20) and (320,24.48) .. (320,30) .. controls (320,35.52) and (315.52,40) .. (310,40) .. controls (304.48,40) and (300,35.52) .. (300,30) -- cycle ;
\draw  [fill={rgb, 255:red, 0; green, 0; blue, 0 }  ,fill opacity=1 ] (120,250) .. controls (120,244.48) and (124.48,240) .. (130,240) .. controls (135.52,240) and (140,244.48) .. (140,250) .. controls (140,255.52) and (135.52,260) .. (130,260) .. controls (124.48,260) and (120,255.52) .. (120,250) -- cycle ;
\draw  [fill={rgb, 255:red, 0; green, 0; blue, 0 }  ,fill opacity=1 ] (360,250) .. controls (360,244.48) and (364.48,240) .. (370,240) .. controls (375.52,240) and (380,244.48) .. (380,250) .. controls (380,255.52) and (375.52,260) .. (370,260) .. controls (364.48,260) and (360,255.52) .. (360,250) -- cycle ;
\draw  [fill={rgb, 255:red, 0; green, 0; blue, 0 }  ,fill opacity=1 ] (240,250) .. controls (240,244.48) and (244.48,240) .. (250,240) .. controls (255.52,240) and (260,244.48) .. (260,250) .. controls (260,255.52) and (255.52,260) .. (250,260) .. controls (244.48,260) and (240,255.52) .. (240,250) -- cycle ;
\draw  [fill={rgb, 255:red, 0; green, 0; blue, 0 }  ,fill opacity=1 ] (480,250) .. controls (480,244.48) and (484.48,240) .. (490,240) .. controls (495.52,240) and (500,244.48) .. (500,250) .. controls (500,255.52) and (495.52,260) .. (490,260) .. controls (484.48,260) and (480,255.52) .. (480,250) -- cycle ;
\draw  [fill={rgb, 255:red, 50; green, 211; blue, 33 }  ,fill opacity=1 ] (272,80) .. controls (272,75.58) and (275.58,72) .. (280,72) .. controls (284.42,72) and (288,75.58) .. (288,80) .. controls (288,84.42) and (284.42,88) .. (280,88) .. controls (275.58,88) and (272,84.42) .. (272,80) -- cycle ;
\draw  [fill={rgb, 255:red, 50; green, 211; blue, 33 }  ,fill opacity=1 ] (302,80) .. controls (302,75.58) and (305.58,72) .. (310,72) .. controls (314.42,72) and (318,75.58) .. (318,80) .. controls (318,84.42) and (314.42,88) .. (310,88) .. controls (305.58,88) and (302,84.42) .. (302,80) -- cycle ;
\draw  [fill={rgb, 255:red, 50; green, 211; blue, 33 }  ,fill opacity=1 ] (332,80) .. controls (332,75.58) and (335.58,72) .. (340,72) .. controls (344.42,72) and (348,75.58) .. (348,80) .. controls (348,84.42) and (344.42,88) .. (340,88) .. controls (335.58,88) and (332,84.42) .. (332,80) -- cycle ;

\draw (245,14) node [anchor=north west][inner sep=0.75pt]  [font=\small] [align=left] {
$
S( 1,\ 4)
$};
\draw (193,138) node [anchor=north west][inner sep=0.75pt]  [font=\small] [align=left] {
$
S( 1,\ 3)
$};
\draw (377,138) node [anchor=north west][inner sep=0.75pt]  [font=\small] [align=left] {
$
S( 2,\ 4)
$};
\draw (161,206) node [anchor=north west][inner sep=0.75pt]  [font=\small] [align=left] {
$
S( 1,\ 2)
$};
\draw (285,150) node [anchor=north west][inner sep=0.75pt]  [font=\small] [align=left] {
$
S( 2,\ 3)
$};
\draw (404,202) node [anchor=north west][inner sep=0.75pt]  [font=\small] [align=left] {
$
S( 3,\ 4)
$};
\draw (104,262) node [anchor=north west][inner sep=0.75pt]  [font=\small] [align=left] {
$
S( 1,\ 1)
$};
\draw (344,262) node [anchor=north west][inner sep=0.75pt]  [font=\small] [align=left] {
$
S( 3,\ 3)
$};
\draw (468,262) node [anchor=north west][inner sep=0.75pt]  [font=\small] [align=left] {
$
S( 4,\ 4)
$};
\draw (225,262) node [anchor=north west][inner sep=0.75pt]  [font=\small] [align=left] {
$
S( 2,\ 2)
$};

\end{tikzpicture} 
        }
        \caption{Computation of subproblems using \textit{leaf vertices}. Root vertices are in black and leaves belonging to subproblem $S\mybraces{1, 4}$ are in shown in green.}
        \label{fig:fulldag}
    \end{subfigure}
    \caption{Construction of CDAG in $\mathcal{G}\mybraces{n}$}
    \label{fig:cdag-construction}
\end{figure}
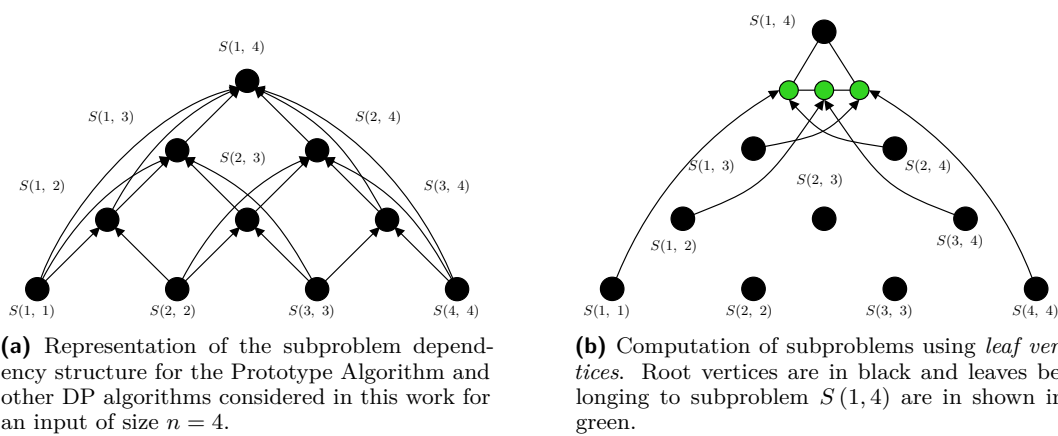

In this section, we give a construction for a family of CDAGs representing the execution of a class of DP algorithms exhibiting a subproblem optimality structure as the one outlined for the Prototype Algorithm~\ref{alg:proto}. 

Given input size $n$, the CDAG $G$ corresponding to the execution of $\mathcal{A}^*$ is constructed using the CDAG in Figure~\ref{fig:simpledag} as the base structure. For each vertex in this CDAG corresponding to subproblem $S\mybraces{i,j}$ we add to $G$ a directed binary tree DAG with $j-i$ leaves with all edges directed towards the root.  The `\emph{root vertex} ($R$-vertex for short) of the tree DAG, henceforth referred to as $v_{i,j}$, corresponds to the computation of the solution of the associated subproblem $S(i,j)$. 
The $k$-th ``\emph{leaf vertex}'' ($L$-vertex for short) of the tree, for $k\in\{0,\ldots,j-i-1\}$  results from the combination of subproblems $S(i,i+k)$ and $S(i+k+1,j)$ and has as predecessors the $R$-vertices of the tree DAGs corresponding respectively to $S(i,i+k)$ and $S(i+k+1,j)$ (\myIe{} $v_{i, i+k}$ and $v_{i+k + 1,j})$ (a visual representation is provided in Figure~\ref{fig:fulldag}). Notice that binary trees corresponding to different R-vertices are vertex disjoint. For each such tree DAG, we say that its $j - i$  $L$-vertices \emph{belong to the root} $v_{i,j}$.

Our lower bound results hold for \emph{any possible structure} of these tree DAGs. This is a \emph{feature} of our model meant to accommodate a variety of possible implementations for $\mathcal{A}^*$ and other DP algorithms sharing the same substructure.

As an example, for the classic DP algorithm for the Matrix Chain Multiplication problem presented in Algorithm~\ref{alg:MCM}, the root vertex of each tree sub-CDAG $S\mybraces{i,j}$ corresponds to the computation of the minimum number of operations required to compute the chain product of the matrices from the $i$-th to the $j$-th, and the $k$-th leaf, for $k\in\{0,\ldots,j-i-1\}$, corresponds to the computation of $S\mybraces{i,i+k}+S\mybraces{i+k+1,j}+ d_{i - 1}d_{i+k}d_j$ (see line 10 of Algorithm~\ref{alg:MCM}).

We denote the family of CDAGs constructed in such a way with $n$ input vertices as $\mathcal{G}(n)$. This family of CDAGs captures the dependence structure between subproblems which is common to the family of DP algorithms which are the focus of our analysis.\\

\noindent\textbf{Analysis of CDAG properties:}
Given a CDAG $G$ from the family $\mathcal{G}\mybraces{n}$,
let $R$ (resp. $L$) denote the set its of $R$-vertices (resp. $L$-vertices). Then, $|R| = \frac{n(n+1)}{2}$ and $|L| = \frac{n^3-n}{6}$.

The ``$i$\emph{-th row}'' (resp., the ``$j$\emph{-th column}'') of $G$, written $r_i$ (resp., $c_j$), is defined as the subset of $R$ including all of the $R$-vertices $v_{i,j}$ of the tree sub-CDAGs each corresponding to the computation of subproblem $S(i,j)$ for $j\in \{i,\ldots,n\}$ (resp., $i\in \{1,\ldots,j\}$).
For any $v\in R$, if $v\in r_i $ and $v\in c_j$ then $i\leq j$.  By construction, the rows $(r_1,r_2,\ldots,r_n)$ (resp., the columns $(c_1,c_2,\ldots,c_n)$) partition $R$, and $|r_i|= n-i +1$ (resp., $|c_i|=i$) for $1\leq i\leq n$. A visual representation of rows and columns can be found in Figure~\ref{fig:rowscols}

\begin{figure}
    \centering
    \begin{subfigure}[t]{0.48\textwidth}
        \centering
        \resizebox{\textwidth}{!}{%
            \tikzset{every picture/.style={line width=0.75pt}} 

\begin{tikzpicture}[x=0.75pt,y=0.75pt,yscale=-1,xscale=1]

\draw  [color={rgb, 255:red, 208; green, 2; blue, 27 }  ,draw opacity=1 ][dash pattern={on 4.5pt off 4.5pt}][line width=0.75]  (206.02,173.98) .. controls (208.22,171.78) and (211.78,171.78) .. (213.98,173.98) -- (286.02,246.02) .. controls (288.22,248.22) and (288.22,251.78) .. (286.02,253.98) -- (274.07,265.93) .. controls (271.87,268.13) and (268.31,268.13) .. (266.11,265.93) -- (194.07,193.89) .. controls (191.87,191.69) and (191.87,188.13) .. (194.07,185.93) -- cycle ;
\draw  [fill={rgb, 255:red, 0; green, 0; blue, 0 }  ,fill opacity=1 ] (201,190) .. controls (201,184.48) and (205.48,180) .. (211,180) .. controls (216.52,180) and (221,184.48) .. (221,190) .. controls (221,195.52) and (216.52,200) .. (211,200) .. controls (205.48,200) and (201,195.52) .. (201,190) -- cycle ;
\draw  [fill={rgb, 255:red, 0; green, 0; blue, 0 }  ,fill opacity=1 ] (261,130) .. controls (261,124.48) and (265.48,120) .. (271,120) .. controls (276.52,120) and (281,124.48) .. (281,130) .. controls (281,135.52) and (276.52,140) .. (271,140) .. controls (265.48,140) and (261,135.52) .. (261,130) -- cycle ;
\draw  [fill={rgb, 255:red, 0; green, 0; blue, 0 }  ,fill opacity=1 ] (321,190) .. controls (321,184.48) and (325.48,180) .. (331,180) .. controls (336.52,180) and (341,184.48) .. (341,190) .. controls (341,195.52) and (336.52,200) .. (331,200) .. controls (325.48,200) and (321,195.52) .. (321,190) -- cycle ;
\draw  [fill={rgb, 255:red, 0; green, 0; blue, 0 }  ,fill opacity=1 ] (441,190) .. controls (441,184.48) and (445.48,180) .. (451,180) .. controls (456.52,180) and (461,184.48) .. (461,190) .. controls (461,195.52) and (456.52,200) .. (451,200) .. controls (445.48,200) and (441,195.52) .. (441,190) -- cycle ;
\draw  [fill={rgb, 255:red, 0; green, 0; blue, 0 }  ,fill opacity=1 ] (381,130) .. controls (381,124.48) and (385.48,120) .. (391,120) .. controls (396.52,120) and (401,124.48) .. (401,130) .. controls (401,135.52) and (396.52,140) .. (391,140) .. controls (385.48,140) and (381,135.52) .. (381,130) -- cycle ;
\draw  [fill={rgb, 255:red, 0; green, 0; blue, 0 }  ,fill opacity=1 ] (321,70) .. controls (321,64.48) and (325.48,60) .. (331,60) .. controls (336.52,60) and (341,64.48) .. (341,70) .. controls (341,75.52) and (336.52,80) .. (331,80) .. controls (325.48,80) and (321,75.52) .. (321,70) -- cycle ;
\draw  [fill={rgb, 255:red, 0; green, 0; blue, 0 }  ,fill opacity=1 ] (141,250) .. controls (141,244.48) and (145.48,240) .. (151,240) .. controls (156.52,240) and (161,244.48) .. (161,250) .. controls (161,255.52) and (156.52,260) .. (151,260) .. controls (145.48,260) and (141,255.52) .. (141,250) -- cycle ;
\draw  [fill={rgb, 255:red, 0; green, 0; blue, 0 }  ,fill opacity=1 ] (381,250) .. controls (381,244.48) and (385.48,240) .. (391,240) .. controls (396.52,240) and (401,244.48) .. (401,250) .. controls (401,255.52) and (396.52,260) .. (391,260) .. controls (385.48,260) and (381,255.52) .. (381,250) -- cycle ;
\draw  [fill={rgb, 255:red, 0; green, 0; blue, 0 }  ,fill opacity=1 ] (261,250) .. controls (261,244.48) and (265.48,240) .. (271,240) .. controls (276.52,240) and (281,244.48) .. (281,250) .. controls (281,255.52) and (276.52,260) .. (271,260) .. controls (265.48,260) and (261,255.52) .. (261,250) -- cycle ;
\draw  [fill={rgb, 255:red, 0; green, 0; blue, 0 }  ,fill opacity=1 ] (501,250) .. controls (501,244.48) and (505.48,240) .. (511,240) .. controls (516.52,240) and (521,244.48) .. (521,250) .. controls (521,255.52) and (516.52,260) .. (511,260) .. controls (505.48,260) and (501,255.52) .. (501,250) -- cycle ;
\draw    (151,250) -- (202.88,198.12) ;
\draw [shift={(205,196)}, rotate = 135] [fill={rgb, 255:red, 0; green, 0; blue, 0 }  ][line width=0.08]  [draw opacity=0] (8.93,-4.29) -- (0,0) -- (8.93,4.29) -- cycle    ;
\draw    (271,130) -- (322.88,78.12) ;
\draw [shift={(325,76)}, rotate = 135] [fill={rgb, 255:red, 0; green, 0; blue, 0 }  ][line width=0.08]  [draw opacity=0] (8.93,-4.29) -- (0,0) -- (8.93,4.29) -- cycle    ;
\draw    (215,186) -- (262.88,138.12) ;
\draw [shift={(265,136)}, rotate = 135] [fill={rgb, 255:red, 0; green, 0; blue, 0 }  ][line width=0.08]  [draw opacity=0] (8.93,-4.29) -- (0,0) -- (8.93,4.29) -- cycle    ;
\draw    (331,190) -- (382.88,138.12) ;
\draw [shift={(385,136)}, rotate = 135] [fill={rgb, 255:red, 0; green, 0; blue, 0 }  ][line width=0.08]  [draw opacity=0] (8.93,-4.29) -- (0,0) -- (8.93,4.29) -- cycle    ;
\draw    (271,250) -- (322.88,198.12) ;
\draw [shift={(325,196)}, rotate = 135] [fill={rgb, 255:red, 0; green, 0; blue, 0 }  ][line width=0.08]  [draw opacity=0] (8.93,-4.29) -- (0,0) -- (8.93,4.29) -- cycle    ;
\draw    (271,250) -- (219.12,198.12) ;
\draw [shift={(217,196)}, rotate = 45] [fill={rgb, 255:red, 0; green, 0; blue, 0 }  ][line width=0.08]  [draw opacity=0] (8.93,-4.29) -- (0,0) -- (8.93,4.29) -- cycle    ;
\draw    (391,250) -- (442.88,198.12) ;
\draw [shift={(445,196)}, rotate = 135] [fill={rgb, 255:red, 0; green, 0; blue, 0 }  ][line width=0.08]  [draw opacity=0] (8.93,-4.29) -- (0,0) -- (8.93,4.29) -- cycle    ;
\draw    (331,190) -- (279.12,138.12) ;
\draw [shift={(277,136)}, rotate = 45] [fill={rgb, 255:red, 0; green, 0; blue, 0 }  ][line width=0.08]  [draw opacity=0] (8.93,-4.29) -- (0,0) -- (8.93,4.29) -- cycle    ;
\draw    (391,250) -- (339.12,198.12) ;
\draw [shift={(337,196)}, rotate = 45] [fill={rgb, 255:red, 0; green, 0; blue, 0 }  ][line width=0.08]  [draw opacity=0] (8.93,-4.29) -- (0,0) -- (8.93,4.29) -- cycle    ;
\draw    (451,190) -- (399.12,138.12) ;
\draw [shift={(397,136)}, rotate = 45] [fill={rgb, 255:red, 0; green, 0; blue, 0 }  ][line width=0.08]  [draw opacity=0] (8.93,-4.29) -- (0,0) -- (8.93,4.29) -- cycle    ;
\draw    (391,130) -- (339.12,78.12) ;
\draw [shift={(337,76)}, rotate = 45] [fill={rgb, 255:red, 0; green, 0; blue, 0 }  ][line width=0.08]  [draw opacity=0] (8.93,-4.29) -- (0,0) -- (8.93,4.29) -- cycle    ;
\draw    (511,250) -- (459.12,198.12) ;
\draw [shift={(457,196)}, rotate = 45] [fill={rgb, 255:red, 0; green, 0; blue, 0 }  ][line width=0.08]  [draw opacity=0] (8.93,-4.29) -- (0,0) -- (8.93,4.29) -- cycle    ;
\draw    (151,250) .. controls (176.61,197.3) and (215.8,154.79) .. (262.85,136.8) ;
\draw [shift={(265,136)}, rotate = 159.97] [fill={rgb, 255:red, 0; green, 0; blue, 0 }  ][line width=0.08]  [draw opacity=0] (8.93,-4.29) -- (0,0) -- (8.93,4.29) -- cycle    ;
\draw    (511,250) .. controls (478.5,189.42) and (456.66,163.29) .. (399.63,137.19) ;
\draw [shift={(397,136)}, rotate = 24.19] [fill={rgb, 255:red, 0; green, 0; blue, 0 }  ][line width=0.08]  [draw opacity=0] (8.93,-4.29) -- (0,0) -- (8.93,4.29) -- cycle    ;
\draw    (451,190) .. controls (418.5,129.42) and (396.66,103.29) .. (339.63,77.19) ;
\draw [shift={(337,76)}, rotate = 24.19] [fill={rgb, 255:red, 0; green, 0; blue, 0 }  ][line width=0.08]  [draw opacity=0] (8.93,-4.29) -- (0,0) -- (8.93,4.29) -- cycle    ;
\draw    (391,250) .. controls (358.5,189.42) and (336.66,163.29) .. (279.63,137.19) ;
\draw [shift={(277,136)}, rotate = 24.19] [fill={rgb, 255:red, 0; green, 0; blue, 0 }  ][line width=0.08]  [draw opacity=0] (8.93,-4.29) -- (0,0) -- (8.93,4.29) -- cycle    ;
\draw    (211,190) .. controls (236.61,137.3) and (275.8,94.79) .. (322.85,76.8) ;
\draw [shift={(325,76)}, rotate = 159.97] [fill={rgb, 255:red, 0; green, 0; blue, 0 }  ][line width=0.08]  [draw opacity=0] (8.93,-4.29) -- (0,0) -- (8.93,4.29) -- cycle    ;
\draw    (271,250) .. controls (296.61,197.3) and (335.8,154.79) .. (382.85,136.8) ;
\draw [shift={(385,136)}, rotate = 159.97] [fill={rgb, 255:red, 0; green, 0; blue, 0 }  ][line width=0.08]  [draw opacity=0] (8.93,-4.29) -- (0,0) -- (8.93,4.29) -- cycle    ;
\draw    (151,250) .. controls (168.82,187.13) and (248.39,90.46) .. (322.75,76.4) ;
\draw [shift={(325,76)}, rotate = 170.54] [fill={rgb, 255:red, 0; green, 0; blue, 0 }  ][line width=0.08]  [draw opacity=0] (8.93,-4.29) -- (0,0) -- (8.93,4.29) -- cycle    ;
\draw    (511,250) .. controls (480.31,167.33) and (413.36,96.92) .. (339.25,76.6) ;
\draw [shift={(337,76)}, rotate = 14.57] [fill={rgb, 255:red, 0; green, 0; blue, 0 }  ][line width=0.08]  [draw opacity=0] (8.93,-4.29) -- (0,0) -- (8.93,4.29) -- cycle    ;
\draw  [color={rgb, 255:red, 208; green, 2; blue, 27 }  ,draw opacity=1 ][dash pattern={on 4.5pt off 4.5pt}][line width=0.75]  (266.19,113.85) .. controls (268.39,111.66) and (271.94,111.67) .. (274.13,113.86) -- (406.2,245.93) .. controls (408.39,248.12) and (408.39,251.66) .. (406.19,253.85) -- (394.26,265.7) .. controls (392.07,267.89) and (388.51,267.88) .. (386.32,265.69) -- (254.25,133.62) .. controls (252.06,131.43) and (252.07,127.89) .. (254.26,125.7) -- cycle ;
\draw  [color={rgb, 255:red, 208; green, 2; blue, 27 }  ,draw opacity=1 ][dash pattern={on 4.5pt off 4.5pt}][line width=0.75]  (326.62,54.53) .. controls (328.82,52.35) and (332.37,52.35) .. (334.56,54.54) -- (526.08,246.06) .. controls (528.27,248.25) and (528.26,251.79) .. (526.07,253.97) -- (514.13,265.83) .. controls (511.94,268.01) and (508.38,268.01) .. (506.19,265.82) -- (314.68,74.3) .. controls (312.49,72.11) and (312.5,68.57) .. (314.69,66.39) -- cycle ;
\draw  [color={rgb, 255:red, 208; green, 2; blue, 27 }  ,draw opacity=1 ][dash pattern={on 4.5pt off 4.5pt}][line width=0.75]  (144.02,231.98) .. controls (146.22,229.78) and (149.78,229.78) .. (151.98,231.98) -- (168.02,248.02) .. controls (170.22,250.22) and (170.22,253.78) .. (168.02,255.98) -- (156.07,267.93) .. controls (153.87,270.13) and (150.31,270.13) .. (148.11,267.93) -- (132.07,251.89) .. controls (129.87,249.69) and (129.87,246.13) .. (132.07,243.93) -- cycle ;
\draw  [color={rgb, 255:red, 245; green, 166; blue, 35 }  ,draw opacity=1 ][dash pattern={on 4.5pt off 4.5pt}][line width=0.75]  (346.42,66.47) .. controls (348.6,68.67) and (348.6,72.22) .. (346.41,74.41) -- (154.34,266.48) .. controls (152.15,268.67) and (148.61,268.66) .. (146.43,266.47) -- (134.57,254.54) .. controls (132.39,252.34) and (132.39,248.78) .. (134.58,246.6) -- (326.65,54.53) .. controls (328.84,52.34) and (332.38,52.35) .. (334.56,54.54) -- cycle ;
\draw  [color={rgb, 255:red, 245; green, 166; blue, 35 }  ,draw opacity=1 ][dash pattern={on 4.5pt off 4.5pt}][line width=0.75]  (406.43,126.47) .. controls (408.61,128.66) and (408.6,132.22) .. (406.41,134.41) -- (274.34,266.48) .. controls (272.15,268.67) and (268.61,268.66) .. (266.43,266.47) -- (254.57,254.54) .. controls (252.39,252.34) and (252.39,248.78) .. (254.58,246.6) -- (386.65,114.52) .. controls (388.84,112.33) and (392.39,112.34) .. (394.57,114.54) -- cycle ;
\draw  [color={rgb, 255:red, 245; green, 166; blue, 35 }  ,draw opacity=1 ][dash pattern={on 4.5pt off 4.5pt}][line width=0.75]  (466,186) .. controls (468.2,188.2) and (468.2,191.76) .. (466,193.96) -- (394.26,265.7) .. controls (392.06,267.9) and (388.5,267.9) .. (386.3,265.7) -- (374.35,253.76) .. controls (372.15,251.56) and (372.15,247.99) .. (374.35,245.8) -- (446.09,174.06) .. controls (448.29,171.86) and (451.86,171.86) .. (454.06,174.06) -- cycle ;
\draw  [color={rgb, 255:red, 245; green, 166; blue, 35 }  ,draw opacity=1 ][dash pattern={on 4.5pt off 4.5pt}][line width=0.75]  (528.97,243.97) .. controls (531.17,246.17) and (531.17,249.74) .. (528.97,251.94) -- (512.94,267.97) .. controls (510.74,270.17) and (507.17,270.17) .. (504.97,267.97) -- (493.03,256.03) .. controls (490.83,253.83) and (490.83,250.26) .. (493.03,248.06) -- (509.06,232.03) .. controls (511.26,229.83) and (514.83,229.83) .. (517.03,232.03) -- cycle ;

\draw (348,46.4) node [anchor=north west][inner sep=0.75pt]    {$v_{1,4}$};
\draw (296,119.4) node [anchor=north west][inner sep=0.75pt]    {$v_{1,3}$};
\draw (243,263.4) node [anchor=north west][inner sep=0.75pt]  [color={rgb, 255:red, 245; green, 166; blue, 35 }  ,opacity=1 ]  {$r_{2}$};
\draw (363,263.4) node [anchor=north west][inner sep=0.75pt]  [color={rgb, 255:red, 245; green, 166; blue, 35 }  ,opacity=1 ]  {$r_{3}$};
\draw (123,263.4) node [anchor=north west][inner sep=0.75pt]  [color={rgb, 255:red, 245; green, 166; blue, 35 }  ,opacity=1 ]  {$r_{1}$};
\draw (483,262.4) node [anchor=north west][inner sep=0.75pt]  [color={rgb, 255:red, 245; green, 166; blue, 35 }  ,opacity=1 ]  {$r_{4}$};
\draw (162,263.4) node [anchor=north west][inner sep=0.75pt]  [color={rgb, 255:red, 208; green, 2; blue, 27 }  ,opacity=1 ]  {$c_{1}$};
\draw (282,262.4) node [anchor=north west][inner sep=0.75pt]  [color={rgb, 255:red, 208; green, 2; blue, 27 }  ,opacity=1 ]  {$c_{2}$};
\draw (402,263.4) node [anchor=north west][inner sep=0.75pt]  [color={rgb, 255:red, 208; green, 2; blue, 27 }  ,opacity=1 ]  {$c_{3}$};
\draw (522,262.4) node [anchor=north west][inner sep=0.75pt]  [color={rgb, 255:red, 208; green, 2; blue, 27 }  ,opacity=1 ]  {$c_{4}$};
\draw (176,246.4) node [anchor=north west][inner sep=0.75pt]    {$v_{1,1}$};
\draw (360,183.4) node [anchor=north west][inner sep=0.75pt]    {$v_{2,3}$};
\draw (297,246.4) node [anchor=north west][inner sep=0.75pt]    {$v_{2,2}$};
\draw (417,246.4) node [anchor=north west][inner sep=0.75pt]    {$v_{3,3}$};
\draw (537,247.4) node [anchor=north west][inner sep=0.75pt]    {$v_{4,4}$};
\draw (237,183.4) node [anchor=north west][inner sep=0.75pt]    {$v_{1,2}$};
\draw (476,170.4) node [anchor=north west][inner sep=0.75pt]    {$v_{3,4}$};
\draw (420,106.4) node [anchor=north west][inner sep=0.75pt]    {$v_{2,4}$};

\end{tikzpicture}
        }
        \caption{Rows and columns}
        \label{fig:rowscols}
    \end{subfigure}
    \hfill
    \begin{subfigure}[t]{0.48\textwidth}
        \centering
        \resizebox{\textwidth}{!}{%
            \tikzset{every picture/.style={line width=0.75pt}} 

\begin{tikzpicture}[x=0.75pt,y=0.75pt,yscale=-1,xscale=1]

\draw  [fill={rgb, 255:red, 0; green, 0; blue, 0 }  ,fill opacity=1 ] (102,150) .. controls (102,145.58) and (105.58,142) .. (110,142) .. controls (114.42,142) and (118,145.58) .. (118,150) .. controls (118,154.42) and (114.42,158) .. (110,158) .. controls (105.58,158) and (102,154.42) .. (102,150) -- cycle ;
\draw  [fill={rgb, 255:red, 0; green, 0; blue, 0 }  ,fill opacity=1 ] (162,90) .. controls (162,85.58) and (165.58,82) .. (170,82) .. controls (174.42,82) and (178,85.58) .. (178,90) .. controls (178,94.42) and (174.42,98) .. (170,98) .. controls (165.58,98) and (162,94.42) .. (162,90) -- cycle ;
\draw  [fill={rgb, 255:red, 0; green, 0; blue, 0 }  ,fill opacity=1 ] (222,150) .. controls (222,145.58) and (225.58,142) .. (230,142) .. controls (234.42,142) and (238,145.58) .. (238,150) .. controls (238,154.42) and (234.42,158) .. (230,158) .. controls (225.58,158) and (222,154.42) .. (222,150) -- cycle ;
\draw  [fill={rgb, 255:red, 0; green, 0; blue, 0 }  ,fill opacity=1 ] (342,150) .. controls (342,145.58) and (345.58,142) .. (350,142) .. controls (354.42,142) and (358,145.58) .. (358,150) .. controls (358,154.42) and (354.42,158) .. (350,158) .. controls (345.58,158) and (342,154.42) .. (342,150) -- cycle ;
\draw  [fill={rgb, 255:red, 0; green, 0; blue, 0 }  ,fill opacity=1 ] (282,90) .. controls (282,85.58) and (285.58,82) .. (290,82) .. controls (294.42,82) and (298,85.58) .. (298,90) .. controls (298,94.42) and (294.42,98) .. (290,98) .. controls (285.58,98) and (282,94.42) .. (282,90) -- cycle ;
\draw  [fill={rgb, 255:red, 0; green, 0; blue, 0 }  ,fill opacity=1 ] (42,210) .. controls (42,205.58) and (45.58,202) .. (50,202) .. controls (54.42,202) and (58,205.58) .. (58,210) .. controls (58,214.42) and (54.42,218) .. (50,218) .. controls (45.58,218) and (42,214.42) .. (42,210) -- cycle ;
\draw  [fill={rgb, 255:red, 0; green, 0; blue, 0 }  ,fill opacity=1 ] (282,210) .. controls (282,205.58) and (285.58,202) .. (290,202) .. controls (294.42,202) and (298,205.58) .. (298,210) .. controls (298,214.42) and (294.42,218) .. (290,218) .. controls (285.58,218) and (282,214.42) .. (282,210) -- cycle ;
\draw  [fill={rgb, 255:red, 0; green, 0; blue, 0 }  ,fill opacity=1 ] (162,210) .. controls (162,205.58) and (165.58,202) .. (170,202) .. controls (174.42,202) and (178,205.58) .. (178,210) .. controls (178,214.42) and (174.42,218) .. (170,218) .. controls (165.58,218) and (162,214.42) .. (162,210) -- cycle ;
\draw  [fill={rgb, 255:red, 0; green, 0; blue, 0 }  ,fill opacity=1 ] (402,210) .. controls (402,205.58) and (405.58,202) .. (410,202) .. controls (414.42,202) and (418,205.58) .. (418,210) .. controls (418,214.42) and (414.42,218) .. (410,218) .. controls (405.58,218) and (402,214.42) .. (402,210) -- cycle ;
\draw    (50,210) -- (101.88,158.12) ;
\draw [shift={(104,156)}, rotate = 135] [fill={rgb, 255:red, 0; green, 0; blue, 0 }  ][line width=0.08]  [draw opacity=0] (8.93,-4.29) -- (0,0) -- (8.93,4.29) -- cycle    ;
\draw    (170,90) -- (221.88,38.12) ;
\draw [shift={(224,36)}, rotate = 135] [fill={rgb, 255:red, 0; green, 0; blue, 0 }  ][line width=0.08]  [draw opacity=0] (8.93,-4.29) -- (0,0) -- (8.93,4.29) -- cycle    ;
\draw    (114,146) -- (161.88,98.12) ;
\draw [shift={(164,96)}, rotate = 135] [fill={rgb, 255:red, 0; green, 0; blue, 0 }  ][line width=0.08]  [draw opacity=0] (8.93,-4.29) -- (0,0) -- (8.93,4.29) -- cycle    ;
\draw    (230,150) -- (281.88,98.12) ;
\draw [shift={(284,96)}, rotate = 135] [fill={rgb, 255:red, 0; green, 0; blue, 0 }  ][line width=0.08]  [draw opacity=0] (8.93,-4.29) -- (0,0) -- (8.93,4.29) -- cycle    ;
\draw    (170,210) -- (221.88,158.12) ;
\draw [shift={(224,156)}, rotate = 135] [fill={rgb, 255:red, 0; green, 0; blue, 0 }  ][line width=0.08]  [draw opacity=0] (8.93,-4.29) -- (0,0) -- (8.93,4.29) -- cycle    ;
\draw    (170,210) -- (118.12,158.12) ;
\draw [shift={(116,156)}, rotate = 45] [fill={rgb, 255:red, 0; green, 0; blue, 0 }  ][line width=0.08]  [draw opacity=0] (8.93,-4.29) -- (0,0) -- (8.93,4.29) -- cycle    ;
\draw    (290,210) -- (341.88,158.12) ;
\draw [shift={(344,156)}, rotate = 135] [fill={rgb, 255:red, 0; green, 0; blue, 0 }  ][line width=0.08]  [draw opacity=0] (8.93,-4.29) -- (0,0) -- (8.93,4.29) -- cycle    ;
\draw    (230,150) -- (178.12,98.12) ;
\draw [shift={(176,96)}, rotate = 45] [fill={rgb, 255:red, 0; green, 0; blue, 0 }  ][line width=0.08]  [draw opacity=0] (8.93,-4.29) -- (0,0) -- (8.93,4.29) -- cycle    ;
\draw    (290,210) -- (238.12,158.12) ;
\draw [shift={(236,156)}, rotate = 45] [fill={rgb, 255:red, 0; green, 0; blue, 0 }  ][line width=0.08]  [draw opacity=0] (8.93,-4.29) -- (0,0) -- (8.93,4.29) -- cycle    ;
\draw    (350,150) -- (298.12,98.12) ;
\draw [shift={(296,96)}, rotate = 45] [fill={rgb, 255:red, 0; green, 0; blue, 0 }  ][line width=0.08]  [draw opacity=0] (8.93,-4.29) -- (0,0) -- (8.93,4.29) -- cycle    ;
\draw    (290,90) -- (238.12,38.12) ;
\draw [shift={(236,36)}, rotate = 45] [fill={rgb, 255:red, 0; green, 0; blue, 0 }  ][line width=0.08]  [draw opacity=0] (8.93,-4.29) -- (0,0) -- (8.93,4.29) -- cycle    ;
\draw    (410,210) -- (358.12,158.12) ;
\draw [shift={(356,156)}, rotate = 45] [fill={rgb, 255:red, 0; green, 0; blue, 0 }  ][line width=0.08]  [draw opacity=0] (8.93,-4.29) -- (0,0) -- (8.93,4.29) -- cycle    ;
\draw    (50,210) .. controls (75.61,157.3) and (114.8,114.79) .. (161.85,96.8) ;
\draw [shift={(164,96)}, rotate = 159.97] [fill={rgb, 255:red, 0; green, 0; blue, 0 }  ][line width=0.08]  [draw opacity=0] (8.93,-4.29) -- (0,0) -- (8.93,4.29) -- cycle    ;
\draw    (410,210) .. controls (377.5,149.42) and (355.66,123.29) .. (298.63,97.19) ;
\draw [shift={(296,96)}, rotate = 24.19] [fill={rgb, 255:red, 0; green, 0; blue, 0 }  ][line width=0.08]  [draw opacity=0] (8.93,-4.29) -- (0,0) -- (8.93,4.29) -- cycle    ;
\draw    (350,150) .. controls (317.5,89.42) and (295.66,63.29) .. (238.63,37.19) ;
\draw [shift={(236,36)}, rotate = 24.19] [fill={rgb, 255:red, 0; green, 0; blue, 0 }  ][line width=0.08]  [draw opacity=0] (8.93,-4.29) -- (0,0) -- (8.93,4.29) -- cycle    ;
\draw    (290,210) .. controls (257.5,149.42) and (235.66,123.29) .. (178.63,97.19) ;
\draw [shift={(176,96)}, rotate = 24.19] [fill={rgb, 255:red, 0; green, 0; blue, 0 }  ][line width=0.08]  [draw opacity=0] (8.93,-4.29) -- (0,0) -- (8.93,4.29) -- cycle    ;
\draw    (110,150) .. controls (135.61,97.3) and (174.8,54.79) .. (221.85,36.8) ;
\draw [shift={(224,36)}, rotate = 159.97] [fill={rgb, 255:red, 0; green, 0; blue, 0 }  ][line width=0.08]  [draw opacity=0] (8.93,-4.29) -- (0,0) -- (8.93,4.29) -- cycle    ;
\draw    (170,210) .. controls (195.61,157.3) and (234.8,114.79) .. (281.85,96.8) ;
\draw [shift={(284,96)}, rotate = 159.97] [fill={rgb, 255:red, 0; green, 0; blue, 0 }  ][line width=0.08]  [draw opacity=0] (8.93,-4.29) -- (0,0) -- (8.93,4.29) -- cycle    ;
\draw    (50,210) .. controls (67.82,147.13) and (147.39,50.46) .. (221.75,36.4) ;
\draw [shift={(224,36)}, rotate = 170.54] [fill={rgb, 255:red, 0; green, 0; blue, 0 }  ][line width=0.08]  [draw opacity=0] (8.93,-4.29) -- (0,0) -- (8.93,4.29) -- cycle    ;
\draw    (410,210) .. controls (379.31,127.33) and (312.36,56.92) .. (238.25,36.6) ;
\draw [shift={(236,36)}, rotate = 14.57] [fill={rgb, 255:red, 0; green, 0; blue, 0 }  ][line width=0.08]  [draw opacity=0] (8.93,-4.29) -- (0,0) -- (8.93,4.29) -- cycle    ;
\draw  [color={rgb, 255:red, 208; green, 2; blue, 27 }  ,draw opacity=1 ][dash pattern={on 4.5pt off 4.5pt}] (48,188) -- (108,248) -- (292,68) -- (312,88) -- (108,288) -- (28,208) -- cycle ;
\draw  [color={rgb, 255:red, 74; green, 144; blue, 226 }  ,draw opacity=1 ][dash pattern={on 4.5pt off 4.5pt}] (108,128) -- (228,248) -- (352,128) -- (372,148) -- (228,288) -- (88,148) -- cycle ;
\draw  [fill={rgb, 255:red, 0; green, 0; blue, 0 }  ,fill opacity=1 ] (222,30) .. controls (222,25.58) and (225.58,22) .. (230,22) .. controls (234.42,22) and (238,25.58) .. (238,30) .. controls (238,34.42) and (234.42,38) .. (230,38) .. controls (225.58,38) and (222,34.42) .. (222,30) -- cycle ;
\draw  [color={rgb, 255:red, 248; green, 231; blue, 28 }  ,draw opacity=1 ][dash pattern={on 4.5pt off 4.5pt}] (168,68) -- (348,248) -- (412,188) -- (432,208) -- (348,288) -- (148,88) -- cycle ;

\draw (241,14.4) node [anchor=north west][inner sep=0.75pt]    {$v_{1,4}$};
\draw (188,74.4) node [anchor=north west][inner sep=0.75pt]    {$v_{1,3}$};
\draw (93,258.4) node [anchor=north west][inner sep=0.75pt]  [color={rgb, 255:red, 208; green, 2; blue, 27 }  ,opacity=1 ]  {$W_{1}$};
\draw (217,258.4) node [anchor=north west][inner sep=0.75pt]  [color={rgb, 255:red, 74; green, 144; blue, 226 }  ,opacity=1 ]  {$W_{2}$};
\draw (329,74.4) node [anchor=north west][inner sep=0.75pt]    {$v_{2,4}$};
\draw (385,142.4) node [anchor=north west][inner sep=0.75pt]    {$v_{3,4}$};
\draw (260,140.4) node [anchor=north west][inner sep=0.75pt]    {$v_{2,3}$};
\draw (129,134.4) node [anchor=north west][inner sep=0.75pt]    {$v_{1,2}$};
\draw (69,195.4) node [anchor=north west][inner sep=0.75pt]    {$v_{1,1}$};
\draw (200,203.4) node [anchor=north west][inner sep=0.75pt]    {$v_{2,2}$};
\draw (317,203.4) node [anchor=north west][inner sep=0.75pt]    {$v_{3,\ 3}$};
\draw (433,202.4) node [anchor=north west][inner sep=0.75pt]    {$v_{4,\ 4}$};
\draw (337,254.4) node [anchor=north west][inner sep=0.75pt]  [color={rgb, 255:red, 248; green, 231; blue, 28 }  ,opacity=1 ]  {$W_{3}$};

\end{tikzpicture} 
        }
        \caption{$W$-cover}
        \label{fig:wcover}
    \end{subfigure}
    \caption{Representation of rows, columns, and the $W$-cover for input size $n = 4$.}
    \label{fig:rowsandwcover}
\end{figure}
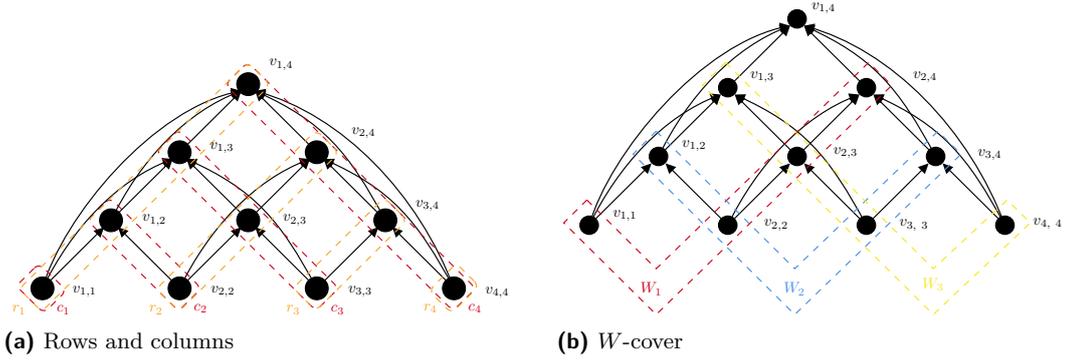


By construction, each $L$-vertex has a distinct pair of $R$-vertices as its predecessors.  Such pair is said to \emph{interact} in the computation. 
Given $v_{i,j}\in R$, it interacts with all and only the vertices $v_{k,i-1}$ for $1\leq k\leq i-1$, and the vertices $v_{j+1,l}$ for $j+1\leq l \leq n$. All such $n+i-j-1$ interactions correspond to the computation of a $L$-vertex belonging to a distinct sub-tree corresponding to the subproblem $S\mybraces{k,j}$ or $S\mybraces{i,l}$ (\myIe{} to the $R$-vertex $v_{k,j}$ or $v_{i,l})$. Pairs of vertices in the same row (or column) do not interact. 

We define $W_i = r_{i+1} \cup c_{i}$, for $1\leq i \leq n-1$ as the subsets of $R$ such that each vertex in $c_i$ interacts with all the vertices of $r_{i+1}$ (a visual representation can be found in Figure~\ref{fig:wcover}). By construction, $|W_i|= n$. A vertex  $v_{1,j}\in r_1$ (resp, $v_{i,n}\in c_n$) is only a member of $W_j$ (resp., $W_{i-1}$). All remaining vertices $v_{i,j}\in R$ are members of both $W_j$ and $W_{i-1}$, which are distinct as, by construction, for any   $v_{i,j}\in R$ we have $i\leq j$. We refer to the collection of sets $\mybraces{W_1,W_2,\ldots,W_{n-1}}$ as the ``\emph{$W$-cover}'' of the $R$-vertices of CDAG $G$. By its definition, and by the properties of the interactions between vertices according to placement in rows and columns discussed in the previous paragraphs, we have the following property:

\begin{lemma}\label{lem:coverproperty}
       Given a CDAG $G\in \mathcal{G}\mybraces{n}$, let $W_1,W_2,\ldots,W_{n-1}$ denote the $W$-cover of its $R$-vertices. For any set $X\subseteq W_i$ there are at most $\frac{|X|2}{4}\leq i\mybraces{n-i}\leq \frac{n^2}{4}$  pairs of vertices $(u,v)$, where $u\in r_{i+1}$ and $v\in c_i$, that interact as predecessors of distinct $L$-vertices each belonging to a distinct $R$-vertex.
\end{lemma}

\section{I/O analysis for computations of CDAGs} \label{sec:main}

We analyze the \io{} complexity of CDAGs in the class $\mathcal{G}(n)$ according to the \emph{Red-Blue Pebble Game} by Hong and Kung~\cite{jia1981complexity}. To simplify the presentation, we assume that $n$ and $M$ are powers of $2$, \myIe{} $n = 2^i$ and $M = 2^j$ for integers $i, j \geq 3$. If that is not the case, our analysis can be adapted with minor, albeit tedious, modifications resulting in different constant terms.  

\subsection{\io{} lower bound for computations with no recomputation}\label{sec:lwbnr}

In this section, we present a lower bound to the \io{} complexity of algorithms corresponding to CDAGS in $\mathcal{G}\mybraces{n}$, assuming that no intermediate value is ever recomputed. (\myIe{} under the ``\emph{no-recomputation assumption}''). The following property captures the relation between subsets of $L$-vertices and the number of their distinct $R$-vertices predecessors:

\begin{restatable}{lemma}{keylemmanr}\label{lem:keycoverpropertynr}
    Let $L'\subset L$ such that $|L'|=8M^{1.5}$ and the vertices in $L'$ belong to the tree sub-CDAGs of at most $4M$ distinct $R$-vertices. The set $R'\subseteq R$ containing all $R$-vertices that are predecessors of at least one vertex in $L'$ has cardinality at least $4M$.
\end{restatable}
\begin{proof}
    Consider the $W$-cover of $G$ as defined in Section~\ref{sec:DAG}. As any $R$-vertex can be a member of at most two distinct $W_i$'s we have:
    \begin{equation}\label{eq:util1}
        |R'|\geq \frac{1}{2}\sum_{i=1}^{n-1} |R'\cap W_i|.
    \end{equation}
    By definition, each vertex in $L$ has a distinct pair of $R$-vertex predecessors which, by construction, are both included in exactly one set $W_i$. Let $\frac{a_i^2}{4}$ be the number of vertices from $L'$ with both predecessors in $W_i$. By the assumptions we have:
    \begin{equation}\label{eq:util2}
        \sum_{i=1}^{n-1} \frac{a_i^2}{4}\geq 8M^{1.5}
    \end{equation}
    
  By Lemma~\ref{lem:coverproperty}, for any $W_i$, each interacting pair of $R$-vertices in it are the predecessors of a distinct $L$-vertex belonging to the tree sub-CDAG of a distinct $R$-vertex. From the assumption that vertices in $L'$ belong to at most $4M$ distinct roots, we have:
    \begin{equation}\label{eq:util3}
        \max_{i=1,\ldots,n-1} \{\frac{a_i^2}{4}\} \leq 4M,
    \end{equation}
    Furthermore, from Lemma~\ref{lem:coverproperty}, we know that vertices in $|R'\cap W_i|$ can interact to produce at most $|R'\cap W_i|^2/4$ vertices in $L'$. Thus,
    
    \begin{equation}\label{eq:util4}
        \sum_{i=1}^{n-1} \frac{|R'\cap W_i|^2}{4}\geq \sum_{i=1}^{n-1} \frac{a_i^2}{4} \Rightarrow\sum_{i=1}^{n-1} |R'\cap W_i|\geq \sum_{i=1}^{n-1} a_i
    \end{equation}
By~\eqref{eq:util1},~\eqref{eq:util4}, ~\eqref{eq:util2},~\eqref{eq:util3}, and Lemma~\ref{lem:mathtrick} (in this order) we have:
\begin{equation*}
    |R'|\geq \frac{1}{2}\sum_{i=1}^{n-1} |R'\cap W_i|\geq \frac{1}{2}\sum_{i=1}^{n-1} a_i \geq \frac{1}{2}\frac{32M^{1.5}}{\sqrt{16M}}\geq 4M. 
\end{equation*}   
\end{proof}
Lemma~\ref{lem:keycoverpropertynr} is a crucial component of our lower bound analysis as it will allow us to claim that in order to compute a set of $L$-vertices, it is necessary to access a number of $R$ vertices, which, due to the non-recomputation assumption, must either be in the cache or loaded into it using a \texttt{read} \io{} operation.

\begin{theorem} \label{thm:lwb}
    For any CDAG $G=(V,E)\in\mathcal{G}(n)$ let $\mathcal{A}$ denote an algorithm corresponding to it. The \io{} complexity of an Algorithm $\mathcal{A}$ when run without recomputing any intermediate value using a cache memory of size $M$ and where for each I/O operation it is possible to move up to $B$ memory words stored in consecutive memory locations from the cache to slow memory or vice versa, is:
    \begin{equation}
        IO_{G}\mybraces{n,M,B}\geq \myMax{\frac{\mybraces{n^3-n}}{16\sqrt{M}}-\frac{n\mybraces{n+1}}{2}-3M,n}\frac{1}{B}
    \end{equation}
\end{theorem}
\begin{proof}
    We prove the result for $B=1$. The result then trivially generalizes for a generic $B$.

The fact that $IO_{G}(n, M,1)\geq n$ follows trivially from our assumption that input is initially stored in slow memory and therefore must be loaded into the cache at least once using $n$ \texttt{read} \io{} operations.

Let $\mathcal{C}$ be any computation schedule for the sequential execution of $\mathcal{A}$ using a cache of size $M$ in which no value corresponding to any vertex of $G$ is computed more than once. We partition $\mathcal{C}$ into non-overlapping segments  $\mathcal{C}_1,\mathcal{C}_2,\ldots$ such that during each $\mathcal{C}_i$ exactly $8M^{1.5}$ values corresponding to distinct L-vertices, denoted as $L_i$, are evaluated. Since $|L| = \frac{n^3}{6}-\frac{n}{6}$  there are $\left\lfloor\mybraces{\frac{n^3}{6}-\frac{n}{6}}\frac{1}{8M^{1.5}}\right\rfloor$ such segments.

 Given $L_i$, we refer to the set of $R$-vertices to whom at least one of the vertices in $L_i$ belongs as $R_i$. These correspond to vertices which are \emph{active} during $\mathcal{C}_i$. For each interval, we denote as $h_i$ the number of $R$-vertices whose value is completely evaluated during $\mathcal{C}_i$. As no value is computed more than once $\sum_{i} h_i=|R|=n\mybraces{n+1}/2$. Let $g_i$ denote the number of \io{} operations executed during $\mathcal{C}_i$.
 We consider two cases:
\begin{itemize}
    \item The vertices of $L_i$ belong to at least $4M+1$ distinct $R$-vertices: For each of the  $4M+1-h_i$ values corresponding to $R$-vertrices that were \textit{partially computed} in $\mathcal{C}_i$, a partial accumulator of the values corresponding to their leaves computed during $\mathcal{C}_i$ must either be in the cache at the end of $\mathcal{C}_i$ or must be written into slow memory by means of a \texttt{write} \io{} operation. Hence, $g_i \geq 4M+1-M-h_i$.
    \item The vertices of $L_i$ belong to at most $4M$ distinct $R$-vertices: By Lemma~\ref{lem:keycoverpropertynr},  the vertices in $L_i$ have at least $4M$  $R$-vertex predecessors, of whom at most $h_i$ are computed during $\mathcal{C}_i$. As in the considered schedule's values are not computed more than once, each of the remaining $4M -h_i$ values corresponding to predecessor vertices of $L_i$ must either be in the cache at the beginning of $\mathcal{C}_i$, or read into the cache by means of a \texttt{}{read} \io{} operation. Hence, $g_i \geq 4M-M-h_i.$
\end{itemize}
    The overall \io{} requirement of the algorithm can be obtained by combining the cost of each non-overlapping segment: 
    \begin{equation*}
        \hspace*{-1.5mm}
        IO_{G}\mybraces{n,M,1} \geq \sum_{i=1}^{\left\lfloor\frac{n^3-n}{48M^{1.5}}\right\rfloor}g_i\geq \sum_{i=1}^{\left\lfloor\frac{n^3-n}{48M^{1.5}}\right\rfloor} 3M -h_i \geq \left\lfloor\frac{n^3-n}{48M^{1.5}}\right\rfloor 3M-\sum_{i=1}^{\frac{n^3-n}{48M^{1.5}}} h_i\geq \frac{\mybraces{n^3-n}}{16\sqrt{M}}-3M-|R|\nonumber.
    \end{equation*}
\end{proof}
By Theorem~\ref{thm:lwb}, we have that for $M\leq cn^2$, for a sufficiently small constant value $c>0$, computations of CDAGs in $\mathcal{G}\mybraces{n}$ in which no value is recomputed require $\Omega\mybraces{\frac{n^3}{\sqrt{M}B}}$ \io{} operations.

\subsection{\io{} lower bound for computations with allowed recomputation}\label{sec:lwbrec}
We extend our \io{} lower bound analysis to general computations in which values \emph{may} be recomputed by combining it with the \emph{dominator} \io{} lower bound technique by Hong and Kung\cite{jia1981complexity}.
\begin{definition}[\cite{jia1981complexity}]\label{def:dominator}
Given $G=(V,E)$, a set $D \subseteq V$ is a \emph{dominator set} for $V'\subseteq V$ if every path directed from any input vertex of $G$ to a vertex in
$V'$ contains a vertex of $\dom$. 
\end{definition}
Dominator sets of $R$-vertices of a CDAG $G\in\mathcal{G}\mybraces{n}$ must satisfy the following property:
\begin{lemma}\label{lem:newtreeinternaldominaotor}
Given $G=(V,E)\in\mathcal{G}(n)$, let $A=\{v_{i_1, j_1},v_{i_2, j_2},\ldots\}$ be a subset of $R$-vertices where  $v_{i_l, j_k}$ is the $R$-vertex corresponding to subproblem $\mathcal{S}(i_k, j_k)$. Any dominator set $D$ of $A$ has minimum cardinality $|D|\geq \min\left(|A|,\min{j_k-i_k}\right)$
\end{lemma}
\begin{proof}



Suppose for every $v_{i, j} \in A$, $D$ contains a vertex internal to the tree-CDAG rooted at $v_{i_k, j_k}$. As the tree-CDAGs of different R-vertices are disjoint, this would imply $|D| \ge |A|$.

Otherwise, there must be some $v_{i, j} \in A$, such that $D$ does not contain any of its internal vertices. Let $L_{i, j} = \{l_0, l_1, \dots, l_{j - i - 1}\}$ denote the set of $L$-vertices corresponding to $v_{i, j}$. It is sufficient to show that there are $j - i$ vertex disjoint paths from the inputs of $G$ to vertices in $L_{i, j}$. By construction, each leaf $l_k \in L_{i, j}$ has a predecessor belonging to a unique column $c_{i+k}$. In turn, each of these predecessors is connected through its internal tree-CDAG to the input vertex in the same column ($v_{i+k, i+k}$). In this manner, $j - i$ paths  $\mybraces{v_{i+k, i+k} \rightarrow v_{i, i+k} \rightarrow l_k}$ for $0 \le k \le j - i - 1$ are obtained. Since each path contains $R$-vertices from a different column, and the tree-CDAGs of different R-vertices are disjoint, the paths outlined above are vertex-disjoint.

\end{proof}

To simplify our analysis, we only consider \emph{parsimonious execution schedules} such that each time an intermediate result is computed, the result is then used to compute at least one of the values of which it is an operand before being removed from the memory (either the cache or slow memory).
Any non-parsimonious schedule $\mathcal{C}$ can be reduced to a parsimonious schedule $\mathcal{C}'$ by removing all the steps that violate the definition of parsimonious computation. $\mathcal{C}'$ has therefore less computational or \io{} operations than $\mathcal{C}$. Hence, restricting the analysis to parsimonious computations leads to no loss of generality.

\begin{theorem} \label{thm:lwbrc}
    For any CDAG $G=(V,E)\in\mathcal{G}(n)$ let $\mathcal{A}$ denote an algorithm corresponding to it. The \io{}-complexity of $\mathcal{A}$ when run using a cache memory of size $M$ and where for each I/O operation it is possible to move up to $B$ memory words stored in consecutive memory locations from cache to slow memory or vice versa, is:
    \begin{equation}
        IO_{G}\mybraces{n,M,B}\geq \mybraces{\myMax{\frac{\mybraces{n-6M-1}^3}{18\sqrt{M}}-2M,n}\frac{1}{B}}
    \end{equation}
\end{theorem}
\begin{proof}
 We prove the result for the case $B=1$. The result then trivially generalizes for a generic $B$. The fact that $IO_{\alg}(n,M,1)\geq n$ follows trivially from our assumption that input is initially stored in slow memory and therefore must be loaded into the cache at least once using $n$ \texttt{read} \io{} operations.

Let $L(x)$ be the set of $L$-vertices in $G$ whose predecessor $R$-vertices both correspond to some subproblems $S\left(i,j\right)$ such that $x\leq j-i$. By the construction, 
for each subproblem $S\left(i,j\right)$ the corresponding tree sub-CDAG has $\max\{0,j-i-2x\}$ such $L$-vertices. Since the sub-CDAGs corresponding to each $S\left(i,j\right)$ do not share vertices, we have:
\begin{align}
    |L\mybraces{x}|&=\sum_{j-i=1}^{n-1}\mybraces{n-\mybraces{j-i}}\myMax{0,j-i-2x}\nonumber\\
    &=\sum_{j-i=2x+1}^{n-1}\mybraces{n-\mybraces{j-i}}\mybraces{j-i-2x}\nonumber\\
     &=\sum_{y=1}^{n-1-2x}\mybraces{n-2x-y}y\nonumber\\
    &= \frac{\mybraces{n-2x-1}\mybraces{n-2x}\mybraces{n-2x+1}}{6}\nonumber
\end{align}
\begin{equation}
        |L\mybraces{x}|> \frac{\mybraces{n-2x-1}^3}{6}
\end{equation}

Let $\mathcal{C}$ be any computation schedule for the sequential execution of $\mathcal{A}$ using a cache of size $M$. We partition $\mathcal{C}$ into segments  $\mathcal{C}_1,\mathcal{C}_2,\ldots$ such that during each $\mathcal{C}_l$ exactly $6M^{1.5}$ values corresponding to distinct vertices in $L\mybraces{3M}$ are computed from their operands (\myIe{} not read from the cache). Let $L_l$ denote the set of vertices corresponding to these values. 
 Below we argue that the number $g_l$ of \io{} operations executed during each of the $\left\lfloor|L\mybraces{3M}|/6M^{1.5}\right\rfloor$ non-overlapping segments $\mathcal{C}_l$ is at least $2M$, from whence the statement follows.
 
Let $A$ denote the set of vertices corresponding to the values computed during $\mathcal{C}_l$. Clearly $L_l\subseteq A$. Let $D$ denote the set of vertices corresponding to either the at most $M$ values that are either stored in the cache or the values that are read into the cache during $\mathcal{C}_l$ by means of \texttt{read} \io{} operations. Thus $g_l= |D|-M$. In order to be possible to compute the values corresponding to vertices in $A$ during $\mathcal{C}_l$ there must be no paths from input vertices of $G$ to vertices in $A$, that is, $D$ must be a dominator set of $A_l$.

 We refer to the set of $R$-vertices to whom at least one of the vertices in $L_l$ belongs as $R_l$. These correspond to vertices which are \emph{active} during $\mathcal{C}_l$. We consider three, mutually exclusive, cases:
\begin{itemize}
    \item[\textbf{(a)}] At least one of the values corresponding to a vertex $v'$ in $R_l$ is \emph{entirely computed} during $\mathcal{C}_l$. That is no accumulator of the values corresponding to the leaves of $v'$ is in the cache or it is loaded in it by means of an \texttt{read} \io{} operation during $\mathcal{C}_l$. Thus, $v'\in A$ and $D$ do not include any non-leaf vertex internal to the tree-sub-CDAG rooted in $v'$. Since $L_l\subseteq L\mybraces{3M}$, by definition, $v$ corresponds to a subproblem $S\mybraces{i,j}$ such that $j-i\geq 3M$. From Lemma~\ref{lem:newtreeinternaldominaotor}, $|D|\geq 3M$ which implies $g_j\geq 2M$.
    \item[\textbf{(b)}] $|R_l|> 4M$: As none of the values in $R_l$ is entirely computed during $\mathcal{C}_l$ (case (a)), and as we are considering parsimonious computations, for each value $R_l$ at least one partial accumulator of its value is either in the cache at the beginning (resp., end) of $\mathcal{C}_l$ or is loaded into the cache during (resp., saved to the slow memory) during it. Hence $g_l\geq 4M-2M$. 
    \item[\textbf{(c)}] $|R_l|\leq 4M$: By a simple modification of Lemma~\ref{lem:keycoverpropertynr},  the vertices in $L_l$ have at least $3M$  distinct $R$-vertices predecessors. Since $L_l\subseteq L\mybraces{3M}$ by definition, all such $R$ vertices corresponds to subproblems $S\mybraces{i,j}$ where $j-i\geq 3M$. As the values of these vertices are either in the cache at the beginning of $\mathcal{C}_l$, or \texttt{read} to the cache during $\mathcal{C}_l$, or computed during $\mathcal{C}_l$, $D$ must be dominator set for the corresponding vertices. From Lemma~\ref{lem:newtreeinternaldominaotor}, we can thus conclude $|D|\geq 3M$ from whence $g_l\geq 2M$.
\end{itemize}
\end{proof}

Thus, by Theorem~\ref{thm:lwbrc}, we have that for $M\leq c_1n$, where $c_1$ is a sufficiently small positive constant value, computations of CDAGs in $\mathcal{G}\mybraces{n}$ require $\Omega\mybraces{\frac{n^3}{\sqrt{M}B}}$ \io{} operations. This is in contrast with the result given for schedules with no recomputation given in Theorem~\ref{thm:lwb} which exhibit \io{} complexity $\Omega\mybraces{\frac{n^3}{\sqrt{M}B}}$ for values of $M$ up to $c_2n^2$ for an appropriately chosen constant $c_2$. 


\subsection{On the tightness of the bounds} \label{sec:tightness}
In \cite{cherng2005cache}, Cherng and Lander present \texttt{Valiant's DP-Closure Algorithm}, which computes algorithms following Prototype Algorithm~\ref{alg:proto} \emph{cache-obliviously} and without recomputation. This Algorithm utilizes a divide and conquer approach, making recursive calls to smaller subproblems (e.g. to compute $S\mybraces{1, n}$, first compute $S\mybraces{1, n/2}$ and so forth). The algorithm also makes recursive calls to two other functions implementing \texttt{Valiant's Star Algorithm} and \texttt{Matrix Multiply and Accumulate}, which use computed previously computed R-vertices to efficiently compute L-vertices. 
A more detailed presentation along with pseudocode can be found in Appendix~\ref{app:valiant}.

When $M < |R| = \frac{n\mybraces{n+1}}{2}$, Valiant's DP Closure algorithm executes $\mathcal{O}\mybraces{\frac{n^3}{ \sqrt{M}B}}$ \io{} operations, and for $M \ge |R|$ it executes $\mathcal{O}(n/B)$ \io{} operations. Thus, we can conclude that our \io{} lower bound for schedules with no recomputation provided in Theorem~\ref{thm:lwb} is asymptotically tight.

\textit{When allowing recomputation}, if $M < c_1n$, where $c_1$ is a sufficiently small positive constant, the \io{} complexity of Valiant's DP-Closure Algorithm asymptotically matches our lower bound in Theorem~\ref{thm:lwbrc}. This implies that for $M < c_1n$, allowing recomputation does not asymptotically change the required number of \io{} operations. On the other hand, if $M > 2n$, our lower bound simplifies to $\Omega(n/B)$, diverging from the complexity of Valiant's DP Closure Algorithm. In Theorem~\ref{thm:largemem-wrc}, we show the existence of an algorithm incurring $\mathcal{O}(n/B)$ \io{} operations for $M > 2n$, verifying the asymptotic tightness of our lower bound. We conclude that for $M > 2n$, allowing the recomputation of intermediate values can result in an asymptotic reduction in the \io{} complexity compared to the no-recomputation case. 

\begin{restatable}[Proof in Appendix~\ref{app:big-mem-recomp}]{theorem}{bigmemio} \label{thm:largemem-wrc} 
    Using a cache memory of size $M \ge 2n$, there exists an algorithm executing $\mathcal{A}^*$ incurring $\mathcal{O}(n/B)$ \io{} operations, where for each I/O operation it is possible to move up to $B$ memory words stored in consecutive memory locations from the cache to slow memory or vice versa.
\end{restatable}

\subsection{Detailed \io{} analysis for selected algorithms} \label{sec:applications}
In Appendix~\ref{app:applications}, we apply the general results in Section~\ref{sec:main} to derive asymptotically tight bounds on the \io{} complexity of matrix chain multiplication, optimal polygon triangulation, and the construction of optimal binary search trees. To do this, we map each of these algorithms onto Prototype Algorithm~\ref{alg:proto}, defining the \texttt{COMBINE}, \texttt{AGGREGATE}, and \texttt{LEAST\_OPTIMAL\_VALUE} appropriately.

\section{\io{} analysis for the Cocke-Younger-Kasami algorithm} \;\label{sec:cyk}
In this section, we present \io{} bounds for the Cocke-Young-Kasami algorithm (CYK). First, we provide some background on the algorithm. We then discuss how our lower bound analysis method must be enhanced to better capture the complexity of this algorithm. Finally, we present an upper bound. 
\subsection{Background on the CYK algorithm}
\begin{definition}[\cite{sipser1996introduction}]\label{def:cnf}
    A Context-Free Grammar is a 4-tuple $\left(V,\mathcal{R}, \Sigma, T\right)$ where
    \begin{itemize}
        \item $V$ is a finite set called the \emph{variables};
        \item $\Sigma$ is a finite set, disjoint from $V$, called the \emph{terminals};
        \item $\mathcal{R}$ is a finite relation in $
 V\times (V\cup \Sigma\cup {\epsilon} )^{*}$, where the $^*$ represents the Kleene star operation and $\epsilon$ denotes the empty string;
        \item $T\in V$ is the start variable. 
    \end{itemize}
\end{definition}

Elements of $\mathcal{R}$ are referred to as \emph{rules} and are generally given in the form $A\rightarrow w$, where $A \in V$, and $w\in(V\cup \Sigma\cup {\epsilon} )^{*}$. If $u,v,w \in (V\cup \Sigma\cup {\epsilon} )^{*}$ and if $A\rightarrow w$ is a rule of the grammar, we say that $uAv$ \emph{yields} $uwv$, written as $uAv\Rightarrow uwv$. We say that $u$ \emph{derives} $v$, written $u\Rightarrow^* v$, if $u=v$ or if there exists a sequence $u_1,u_2,\ldots,u_k$ for $k\geq 1$ such that $u_i\in (V\cup \Sigma\cup {\epsilon} )$ and $u\Rightarrow u_1 \Rightarrow u_2 \Rightarrow \ldots \Rightarrow u_k \Rightarrow v$.
The \emph{language} of a  CFG is the set of all strings $w\in \Sigma^*$ that can be derived from the starting variable $T$, that is $\{w\in \Sigma^*\ s.t.\ T\Rightarrow^* w\}$.

A CFG is in Chomsky's normal form (CNF) if every rule is in the form: $A \rightarrow BC$, or $A \rightarrow a$, or $T\rightarrow \epsilon$,
where $A$ is any variable in $V$, $B$ and $C$ can be any pair in $\left(V\setminus{S}\right)\times \left(V\setminus{S}\right)$, and $a$ can be any terminal in $\Sigma$. Any CFG can be transformed into an equivalent one in CNF through a simple iterative algorithm. We refer to the number of rules $|\mathcal{R}|$ as the \emph{size} of a CFG. We denote the set of binary rules ($A \rightarrow BC$) in a CFG with $\mathcal{R}_B \subset \mathcal{R}$, the set of terminal rules ($A \rightarrow a$) with $\mathcal{R}_T \subset \mathcal{R}$, and the number of distinct right-hand sides in $\mathcal{R}_B$ with $\Gamma$.


Given a CFG $\left(V,\Sigma, R, S\right)$ in CNF, and a string $w=w_1w_2\ldots w_n$ the Cocke–Younger–Kasami  (CYK) algorithm, henceforth referred to as $\mathcal{A}_{CYK}$, decides whether $w$ is part of the language of the CFG. We refer to the standard implementation provided in Algorithm~\ref{alg:CYK}. Let  $S\left(i,j\right)$ be the set of variables that derive the substring $w_i w_{i+1}\ldots w_j$. CYK uses a DP approach to compute all sets $S\mybraces{i,j}$ for $1\leq i\leq j\leq n$  using an adaptation of prototype algorithm $\mathcal{A}^*$.
$\mathcal{A}_{CYK}$ determines whether $w$ is a member of the language generated by the given CFG by checking whether $T\in S\mybraces{1,n}$ (\myIe{} whether $T$ derives $w$).\\

\noindent\textbf{Memory representation model:}
We assume that the variables in $V$ can be referenced by index and that such index can be stored in a single memory word.
 We represent subsets of $V$ (the results of subproblems $S\mybraces{i,j}$)) using a one-hot encoding with $\myceil{|V|/s}$ memory words, where $s$ denotes the number of bits in a memory word. Thus, given the representation of one such subset in the memory, to check whether a variable $V_i$ is a member of the subset it is sufficient to access the $(i\mod s)$-th bit of the $\lfloor i/s\rfloor$-th memory word.
 
We assume that each grammar rule is represented utilizing at most three memory words each used to represent the index of the left-hand side and at most two to represent the right-hand side (either the index of the two variables or the code for a terminal symbol).

\subsection{I/O lower bound}\label{sec:lwbcyk}
While the CYK algorithm shares the underlying structure of relations between subproblems of the prototype algorithm $\mathcal{A}^*$, it presents several non-trivial complications due to the application of the grammar rules that, in turn, require a refinement of our analysis methods and, in particular, of the CDAG construction. In this section, we will assume that the set of grammar rules $\mathcal{R}$ are such that no two rules have the same right-hand side. While this might not always hold in practice, our analysis can be applied by considering a subset of rules of the considered grammar which satisfies this condition.\\
\begin{figure}
    \centering
    \resizebox{0.5\textwidth}{!}{%
\tikzset{every picture/.style={line width=0.75pt}} 

\begin{tikzpicture}[x=0.75pt,y=0.75pt,yscale=-1,xscale=1]

\draw  [fill={rgb, 255:red, 0; green, 0; blue, 0 }  ,fill opacity=1 ] (360,210) .. controls (360,204.48) and (364.48,200) .. (370,200) .. controls (375.52,200) and (380,204.48) .. (380,210) .. controls (380,215.52) and (375.52,220) .. (370,220) .. controls (364.48,220) and (360,215.52) .. (360,210) -- cycle ;
\draw  [fill={rgb, 255:red, 0; green, 0; blue, 0 }  ,fill opacity=1 ] (280,210) .. controls (280,204.48) and (284.48,200) .. (290,200) .. controls (295.52,200) and (300,204.48) .. (300,210) .. controls (300,215.52) and (295.52,220) .. (290,220) .. controls (284.48,220) and (280,215.52) .. (280,210) -- cycle ;
\draw  [fill={rgb, 255:red, 0; green, 0; blue, 0 }  ,fill opacity=1 ] (200,210) .. controls (200,204.48) and (204.48,200) .. (210,200) .. controls (215.52,200) and (220,204.48) .. (220,210) .. controls (220,215.52) and (215.52,220) .. (210,220) .. controls (204.48,220) and (200,215.52) .. (200,210) -- cycle ;
\draw   (180,210) .. controls (180,193.29) and (229.25,179.75) .. (290,179.75) .. controls (350.75,179.75) and (400,193.29) .. (400,210) .. controls (400,226.71) and (350.75,240.25) .. (290,240.25) .. controls (229.25,240.25) and (180,226.71) .. (180,210) -- cycle ;
\draw   (210,210) -- (246,280) -- (176,280) -- cycle ;
\draw   (370,210) -- (406,280) -- (336,280) -- cycle ;
\draw  [fill={rgb, 255:red, 189; green, 16; blue, 224 }  ,fill opacity=1 ] (238,280) .. controls (238,275.58) and (241.58,272) .. (246,272) .. controls (250.42,272) and (254,275.58) .. (254,280) .. controls (254,284.42) and (250.42,288) .. (246,288) .. controls (241.58,288) and (238,284.42) .. (238,280) -- cycle ;
\draw  [fill={rgb, 255:red, 189; green, 16; blue, 224 }  ,fill opacity=1 ] (204,280) .. controls (204,275.58) and (207.58,272) .. (212,272) .. controls (216.42,272) and (220,275.58) .. (220,280) .. controls (220,284.42) and (216.42,288) .. (212,288) .. controls (207.58,288) and (204,284.42) .. (204,280) -- cycle ;
\draw  [fill={rgb, 255:red, 189; green, 16; blue, 224 }  ,fill opacity=1 ] (328,280) .. controls (328,275.58) and (331.58,272) .. (336,272) .. controls (340.42,272) and (344,275.58) .. (344,280) .. controls (344,284.42) and (340.42,288) .. (336,288) .. controls (331.58,288) and (328,284.42) .. (328,280) -- cycle ;
\draw  [fill={rgb, 255:red, 189; green, 16; blue, 224 }  ,fill opacity=1 ] (398,280) .. controls (398,275.58) and (401.58,272) .. (406,272) .. controls (410.42,272) and (414,275.58) .. (414,280) .. controls (414,284.42) and (410.42,288) .. (406,288) .. controls (401.58,288) and (398,284.42) .. (398,280) -- cycle ;
\draw  [fill={rgb, 255:red, 189; green, 16; blue, 224 }  ,fill opacity=1 ] (363,280) .. controls (363,275.58) and (366.58,272) .. (371,272) .. controls (375.42,272) and (379,275.58) .. (379,280) .. controls (379,284.42) and (375.42,288) .. (371,288) .. controls (366.58,288) and (363,284.42) .. (363,280) -- cycle ;
\draw   (176,280) -- (202,340) -- (153,340) -- cycle ;
\draw   (405,280) -- (431,340) -- (382,340) -- cycle ;
\draw  [fill={rgb, 255:red, 42; green, 139; blue, 75 }  ,fill opacity=1 ] (147,340) .. controls (147,336.69) and (149.69,334) .. (153,334) .. controls (156.31,334) and (159,336.69) .. (159,340) .. controls (159,343.31) and (156.31,346) .. (153,346) .. controls (149.69,346) and (147,343.31) .. (147,340) -- cycle ;
\draw  [fill={rgb, 255:red, 42; green, 139; blue, 75 }  ,fill opacity=1 ] (196,340) .. controls (196,336.69) and (198.69,334) .. (202,334) .. controls (205.31,334) and (208,336.69) .. (208,340) .. controls (208,343.31) and (205.31,346) .. (202,346) .. controls (198.69,346) and (196,343.31) .. (196,340) -- cycle ;
\draw  [fill={rgb, 255:red, 42; green, 139; blue, 75 }  ,fill opacity=1 ] (171,340) .. controls (171,336.69) and (173.69,334) .. (177,334) .. controls (180.31,334) and (183,336.69) .. (183,340) .. controls (183,343.31) and (180.31,346) .. (177,346) .. controls (173.69,346) and (171,343.31) .. (171,340) -- cycle ;
\draw  [fill={rgb, 255:red, 42; green, 139; blue, 75 }  ,fill opacity=1 ] (376,340) .. controls (376,336.69) and (378.69,334) .. (382,334) .. controls (385.31,334) and (388,336.69) .. (388,340) .. controls (388,343.31) and (385.31,346) .. (382,346) .. controls (378.69,346) and (376,343.31) .. (376,340) -- cycle ;
\draw  [color={rgb, 255:red, 0; green, 0; blue, 0 }  ,draw opacity=1 ][fill={rgb, 255:red, 42; green, 139; blue, 75 }  ,fill opacity=1 ] (425,340) .. controls (425,336.69) and (427.69,334) .. (431,334) .. controls (434.31,334) and (437,336.69) .. (437,340) .. controls (437,343.31) and (434.31,346) .. (431,346) .. controls (427.69,346) and (425,343.31) .. (425,340) -- cycle ;
\draw  [fill={rgb, 255:red, 42; green, 139; blue, 75 }  ,fill opacity=1 ] (401,340) .. controls (401,336.69) and (403.69,334) .. (407,334) .. controls (410.31,334) and (413,336.69) .. (413,340) .. controls (413,343.31) and (410.31,346) .. (407,346) .. controls (403.69,346) and (401,343.31) .. (401,340) -- cycle ;
\draw  [fill={rgb, 255:red, 189; green, 16; blue, 224 }  ,fill opacity=1 ] (168,280) .. controls (168,275.58) and (171.58,272) .. (176,272) .. controls (180.42,272) and (184,275.58) .. (184,280) .. controls (184,284.42) and (180.42,288) .. (176,288) .. controls (171.58,288) and (168,284.42) .. (168,280) -- cycle ;
\draw  [fill={rgb, 255:red, 189; green, 16; blue, 224 }  ,fill opacity=1 ] (398,280) .. controls (398,275.58) and (401.58,272) .. (406,272) .. controls (410.42,272) and (414,275.58) .. (414,280) .. controls (414,284.42) and (410.42,288) .. (406,288) .. controls (401.58,288) and (398,284.42) .. (398,280) -- cycle ;

\draw (257,150) node [anchor=north west][inner sep=0.75pt]  [font=\small] [align=left] {

$S( i,\ j)$
};
\draw (221,194) node [anchor=north west][inner sep=0.75pt]  [font=\small] [align=left] {
$
A
$};
\draw (301,194) node [anchor=north west][inner sep=0.75pt]  [font=\small] [align=left] {
$
B
$};
\draw (382,194) node [anchor=north west][inner sep=0.75pt]  [font=\small] [align=left] {
$
C
$};
\draw (121,254) node [anchor=north west][inner sep=0.75pt]  [font=\footnotesize] [align=left] {
$
A\rightarrow BD
$};
\draw (227,290) node [anchor=north west][inner sep=0.75pt]  [font=\footnotesize] [align=left] {
$
A\rightarrow DG
$};
\draw (301,290) node [anchor=north west][inner sep=0.75pt]  [font=\footnotesize] [align=left] {
$
C\rightarrow AB
$};
\draw (413,258) node [anchor=north west][inner sep=0.75pt]  [font=\footnotesize] [align=left] {
$
C\rightarrow DE
$};
\draw (121,350) node [anchor=north west][inner sep=0.75pt]  [font=\footnotesize] [align=left] {
$
k\ =\ i\ \ \ i+2\ \ \ j
$};
\draw (349,350) node [anchor=north west][inner sep=0.75pt]  [font=\footnotesize] [align=left] {
$
k\ =\ i\ \ \ i+2\ \ \ j
$};

\end{tikzpicture}
}   
    \caption{Each subproblem $S(i, j)$ is represented by multiple \textit{variable roots} (in black), which is the result of the composition multiple \textit{grammar roots} (in purple), which in turn are the composition of multiple \textit{leaves} (in green).}
    \label{fig:cyk-dag}
\end{figure}
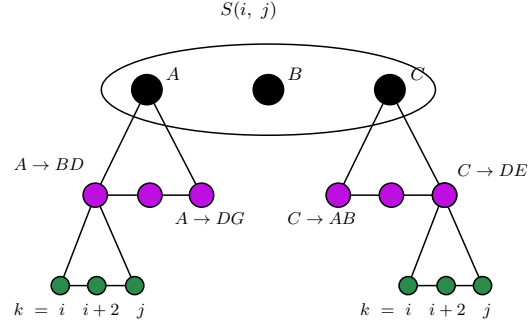

\noindent\textbf{CYK CDAG construction:}\label{sec:cykdagcontruction}
The construction of the CDAG representing CYK's execution for an input string $w$ with $n$ characters and a CFG in CNF $\mybraces{V,\mathcal{R},\Sigma, T}$, called $G_{CYK}\mybraces{n,\mybraces{V,\mathcal{R},\Sigma, T}}$, follows a recursive structure which uses as a basis $G=\mybraces{V,E}\in \mathcal{G}(n)$. We describe the necessary modification focusing on the sub-CDAG for a single $S\mybraces{i,j}$ subproblem. A visual depiction of the ensuing description is provided in Figure~\ref{fig:cyk-dag}:

\begin{itemize}
    \item For the base of the construction, for each $S\mybraces{i, i}$, the CDAG has $|V|/s$ vertices, each corresponding to the encoding of one of the possible variables of the grammar. These vertices serve as the input of the CDAG.
    \item Each $R$-vertex of $G$, corresponding to the computation of the solution of one of the subproblems $S(i,j)$, is replaced by $|V|/s$ vertices, each representing a part of the encoding of the set of variables $V$. 
    We refer to this set of vertices as ``\emph{variable roots}'' $\mathit{VR}$.
    \item Consider a variable root $v$ corresponding to subproblem $S(i,j)$ and variable $A \in V$. $v$ forms the root of a binary tree, which has as its leaves vertices for each rule having $A$ as its left-hand side. We denote the set of these vertices as ``\emph{grammar roots}''  $\mathit{GR}$. Note that as only one variable can appear on the left-hand side of a rule, the trees composed in this way are vertex disjoint. Further, for each $S(i,j))$ there will be exactly $\Gamma$ such vertices. 
    \item For each subproblem $S(i,j)$, consider one of its grammar roots $u$ corresponding to rule $\gamma \in \mathcal{R}_B$. $u$ forms the root of a binary tree combining the results of $j-i$  ``leaf-vertices''. Each leaf vertex corresponds to a combination of the subproblems $S\mybraces{i,k}$ and $S\mybraces{k+1,j}$, for $i\leq k < j$.  Such a vertex has as predecessors the variable root vertices of subproblems $S\mybraces{i,k}$ and $S\mybraces{k+1,j}$ encoding the variables appearing in the right-hand side of rule $\gamma$.
\end{itemize}

Given the CDAG $G_{CYK}\mybraces{n,\mybraces{V,\mathcal{R},\Sigma, T}}$, we generalize the notion of rows (resp., columns) introduced in Section~\ref{sec:DAG} by having row $r_i$ (resp. column $c_j$) include all vertices in $\mathit{VR}$ corresponding to the encoding of the results of subproblems $S(i,j)$ for $1\leq j\leq n$ (resp., for $1\leq i\leq n$).The definition of $W$-cover  generalizes to $G_{CYK}\mybraces{n,\mybraces{V,\mathcal{R},\Sigma, T}}$.
We observe that a property analogous to that outlined in Lemma~\ref{lem:coverproperty} holds for $G_{CYK}\mybraces{n,\mybraces{V,\mathcal{R},\Sigma, T}}$ as well:
\begin{restatable}[Proof in Appendix~\ref{app:cyk-proofs}]{lemma}{cykcover}\label{lem:cykcoverproperty}
\sloppy Given a CDAG $G_{CYK}\mybraces{n,\mybraces{V,\mathcal{R},\Sigma, T}}$, let $W_1,\ldots,W_{n-1}$ denote its $W$-cover.  Any set of $x$ $\mathit{VR}$-vertices in $W_i$ contain at most $x^2/4$ \textit{interacting} pairs, each of which forms the predecessors of a unique  $L$-vertex belonging to a unique $\mathit{GR}$-vertex.
\end{restatable}

\noindent\textbf{\io{} lower bound proof:} 
We analyze the \io{} complexity of any sequential computation for the CYK algorithm by analyzing the CDAG corresponding to the given CFG and the input string to be considered. We first state two results that correspond each to a modification of the results in Lemma~\ref{lem:keycoverpropertynr} and, respectively, of Lemma~\ref{lem:newtreeinternaldominaotor} for the family of CDAG corresponding to execution of the CYK algorithm:
\begin{lemma}\label{lem:cykpreds}
    Given a CDAG $G_{CYK}\mybraces{n,\mybraces{V,\mathcal{R},\Sigma, T}}$, let $L'\subset L$ such that $|L'|=8M^{1.5}$ and the vertices in $L'$ belong to the tree sub-DAGs of at most $4M$ distinct $\mathit{GR}$-vertices. The set $R'\subseteq \mathit{VR}$ includes all $\mathit{VR}$-vertices that are predecessors of at least one vertex in $L'$ has cardinality at least $4M$.
\end{lemma}
The proof follows a reasoning analogous to that used to prove Lemma~\ref{lem:keycoverpropertynr} adapted for the CDAG $G_{CYK}\mybraces{n,\mybraces{V,\mathcal{R},\Sigma, T}}$ by using the property of $L$-vertices of the CYK CDAG given Lemma~\ref{lem:cykcoverproperty} in place of Lemma~\ref{lem:coverproperty}. 

\begin{lemma}\label{lem:dominatorcyk}
Given $G_{CYK}\mybraces{n,\mybraces{V,\mathcal{R},\Sigma, T}}$, let $A=\{v_{i_1, j_1},v_{i_2, j_2},\ldots\}$ be a subset of $GR$-vertices where  $v_{i_l, j_k}$ is the $GR$-vertex corresponding to subproblem $\mathcal{S}(i_k, j_k)$. Any dominator set $D$ of $A$ has cardinality at least $|D|\geq \min\left(|A|,\min{j_k-i_k}\right)$
\end{lemma}
The proof corresponds to the one of Lemma~\ref{lem:newtreeinternaldominaotor} where the tree sub-CDAGs rooted in $GR$-vertices are considered instead of the $R$ vertices of $\mathcal{G}(n)$.

A lower bound to the \io{} complexity of the CYK algorithm can be obtained by following steps analogous to those used in the proof of Theorem~\ref{thm:lwbrc} with opportune adjustments matching the modifications in the \emph{enhanced} CDAG described in this section.

\begin{theorem} \label{thm:cyk-lwb}
    Consider a string $w$ of length $n$  and a CFG $\mybraces{V,\mathcal{R},\Sigma, T}$ in CNF. The number of \io{} operations executed by $\mathcal{A}_{CYK}$ when deciding whether $w$ is a member of the language of the given grammar using a machine equipped with a cache memory of size $M$ is bounded as:
    \begin{equation*}\label{eq:cyklwb}
        IO_{G_{CYK}}\mybraces{n,m,B} \geq \mybraces{\myMax{\frac{\mybraces{n-6M-1}^3}{18\sqrt{M}}\Gamma-2M\Gamma,n}\frac{1}{B}}
    \end{equation*}
        where $\Gamma$  denotes the number of distinct right-hand sides in $\mathcal{R}_B$.
\end{theorem}

By Theorem~\ref{thm:cyk-lwb}, we have that for $M\leq c_1n$, where $c_1$ is a sufficiently small positive constant value, computations of the CYK algorithm require $\Omega\mybraces{\frac{n^3\Gamma}{\sqrt{M}B}}$ \io{} operations. 
Note that if the considered CFG is such that $\Gamma=\Theta\mybraces{|\mathcal{R}_B|}$ our lower bound shows a direct dependence of the \io{} complexity of execution of the CYK algorithm on the \emph{size}  of the considered grammar. Indeed there are many possible grammars for which this is the case. 
Finally, while the result of Theorem~\ref{thm:lwbrc} holds for computations with recomputation, it is also possible to obtain an \io{} lower bound for schedules in which recomputation is not allowed by modifying the analysis given in Theorem~\ref{thm:lwb}.

\subsection{I/O upper bound} \label{sec:cyk-upper-bound}
We present an efficient cache-oblivious implementation of the CYK algorithm, henceforth referred to as $\mathcal{A}_{CYK}$, obtained by refining \texttt{Valiant's DP Closure Algorithm} (presented in Section~\ref{sec:tightness} and Appendix~\ref{app:valiant}).
\newcommand{\matA}{ \begin{bmatrix} X_{11} & X_{12} \\   & X_{22}\end{bmatrix}}
\newcommand{\matB}{ \begin{bmatrix} X_{33} & X_{34} \\   & X_{44}\end{bmatrix}}
\newcommand{\matC}{ \begin{bmatrix} X_{22} & X_{23} \\   & X_{33}\end{bmatrix}}
\newcommand{\matD}{ \begin{bmatrix} X_{11} & X_{13} \\   & X_{33}\end{bmatrix}}
\newcommand{\matE}{ \begin{bmatrix} X_{22} & X_{24} \\   & X_{44}\end{bmatrix}}
\newcommand{\matF}{ \begin{bmatrix} X_{11} & X_{14} \\   & X_{44}\end{bmatrix}}
\begin{algorithm}
\caption{$\mathcal{A}_{CYK^*}$, Valiant's Star Algorithm for CYK} \label{alg:CYK_star}
\begin{algorithmic}[1]
\footnotesize
\State \textbf{Input} $X$ \Comment{Top-left and bottom-right are assumed to be fully computed}
\State \textbf{Output} $X$ \Comment{DP-Closure of X is computed in-place}
\If{$\text{dim}(X) = 2 $} \Comment{Root has been fully computed. Update all variable roots}
    \For{$A \in V$ matched with placeholder $P$}
        \State $X^A \leftarrow X^P$
    \EndFor
    \State \textbf{return}
\EndIf
\State $ \matC \leftarrow \mathcal{A}_{CYK^*}\mybraces{\matC}$
\For{$P \rightarrow BC$ \textbf{in} $\mathcal{R}_B^\Gamma$}
    \State $X_{13}^P \leftarrow X_{13}^P + X_{12}^B \cdot X_{13}^C$ \Comment{Computed using \texttt{Matrix Multiply and Accumulate}}
\EndFor
\State $ \matD \leftarrow \mathcal{A}_{CYK^*}\mybraces{\matD}$
\For{$P \rightarrow BC$ \textbf{in} $\mathcal{R}_B^\Gamma$}
    \State $X_{24}^P \leftarrow X_{24}^P + X_{23}^B \cdot X_{34}^C$
\EndFor
\State $ \matE \leftarrow \mathcal{A}_{CYK^*}\mybraces{\matE}$
\For{$P \rightarrow BC$ \textbf{in} $\mathcal{R}_B^\Gamma$}
    \State $X_{14}^P \leftarrow X_{14}^P + X_{12}^B \cdot X_{24}^C$
    \State $X_{14}^P \leftarrow X_{14}^P + X_{13}^B \cdot X_{34}^C$
\EndFor
\State $ \matF \leftarrow \mathcal{A}_{CYK^*}\mybraces{\matF}$
\end{algorithmic}
\end{algorithm}
To compute the input values corresponding to subproblems $S(i, i)$ for $1\le i \le n$, each character in the input string $w$ is compared with the right-hand side of each terminal rule in $\mathcal{R}_T$, and the appropriate values are stored. The rest of the computation involves executing \texttt{Valiant's DP Closure Algorithm} using a modified version of \texttt{Valiant's Star Algorithm}. Pseudocode for the modified algorithm is provided in Algorithm~\ref{alg:CYK_star}.

Here, $X$ represents a matrix containing values corresponding to $R$-vertices, and $X^A$ represents the corresponding memory words in $X$ containing variable roots for $A \in V$. Our proposed version introduces two modifications of the original algorithm as presented in Appendix~\ref{app:valiant-star}: First, the algorithm iterates over the rules of the grammar with distinct right-hand-sides in lines $8$, $11$, and $14$, calling the subroutine \texttt{Matrix Multiply and Accumulate} (Appendix~\ref{app:matmul}) with the variable roots corresponding to each rule in the grammar. Second, to avoid computing rules with identical right-hand sides more than once, let $\mathcal{R}_B^\Gamma$ denote the largest set of binary rules, each having a unique right-hand side from $\mathcal{R}_B$. In this set, the corresponding left-hand sides are replaced with placeholder variables. For any subproblem, once the value corresponding to a rule in $\mathcal{R}_B^\Gamma$ is fully computed, it can be used to update the value corresponding to all variable roots producing the same right-hand side through a rule in $\mathcal{R}_B$. This is done in lines $4$ and $5$ as the base case implies that all leaves for the subproblem have been computed.

\begin{theorem} \label{thm:cyk-ub}
Given an input string $w$ of length $n$ and a CFG $\mybraces{V,\mathcal{R},\Sigma, T}$ in CNF, algorithm $\mathcal{A}_{CYK}$ decides whether $w$ is a member of the language of the given grammar. When run on a machine equipped with a cache of size $M$, the number of \io{} operations executed by $\mathcal{A}_{CYK}$ can be bound as:
\begin{equation*}
    IO_{\mathcal{A}_{CYK}}\left(n,M\right)\leq \mathcal{O}\mybraces{\mybraces{\frac{n^3\Gamma}{\sqrt{M}} + n^2\Gamma\log{M} + n^2|\mathcal{R}_B| + n|\mathcal{R}_T|}B^{-1}}
\end{equation*}
\end{theorem}
\begin{proof}
    When computing 
    $S\mybraces{i, i}$ for $1 \le i \le n$, $\mathcal{A}_{CYK}$ compares each character of input string $w$ with the right-hand side of each terminal rule in $\mathcal{R}_T$. Doing so incurs at most  $O(n|\mathcal{R}|_T/B)$ \io{} operations. 
    
    In the following, we use $n$ to denote the dimension of the input (\myIe{} $n=dim(X)$).
     The number of \io{} operations  executed by \texttt{Valiant's Star Algorithm} for CYK is at most:
    \begin{equation}\label{eq:cykstarr}
        IO_{\mathcal{A}_{CYK^*}}\mybraces{n} \le
        \begin{cases} 
        \mathcal{O}(\frac{|\mathcal{R}_B|}{B}) & \text{if } n \leq 2, \\
        4 IO_{\mathcal{A}_{CYK^*}}\mybraces{\frac{n}{2}} + 4\Gamma(IO_{MMA}\mybraces{\frac{n}{4}} + O(1))   & \text{otherwise}
        
        \end{cases}
    \end{equation}
     By \cite[Lemma 2.1]{cherng2005cache}, the number of \io{} operations executed by \texttt{Matrix Multiply and Accumulate} is:
    \begin{equation} \label{eq:mma}
        IO_{MMA}\mybraces{n} \le 
        \begin{cases} 
          \mathcal{O}(\frac{n^2}{B}) & \text{if } n^2 \leq M, \\
          \mathcal{O}(\frac{n^3}{B\sqrt{M}}) & \text{otherwise}
        \end{cases}
    \end{equation}
    Solving the recursion in~\eqref{eq:cykstarr}  using~\eqref{eq:mma} yields:
    \begin{align} \label{eq:star-ub}
        IO_{\mathcal{A}_{CYK^*}}\mybraces{n} &\le \mathcal{O}(\frac{n^3\Gamma}{B\sqrt{M}}) + 4^{\log{(n/\sqrt{M})}}\mathcal{O}(\frac{M\Gamma}{B}\log{M} + \frac{M|\mathcal{R}_B|}{B})\nonumber \\
        &\le \mathcal{O}(\frac{n^3\Gamma}{B\sqrt{M}} + \frac{n^2\Gamma}{B}\log{M} + \frac{n^2|\mathcal{R}_B|}{B})
    \end{align}
    Finally, by adapting that analysis in~\cite{cherng2005cache} we can bound the number of \io{} operations executed by the proposed algorithm $\mathcal{A}_{CYK}$ as:
    \begin{equation*}
        IO_{\mathcal{A}_{CYK}}\mybraces{n} \le
        \begin{cases} 
        \mathcal{O}(1) & \text{if } n \leq 2, \\
        2 IO_{\mathcal{A}_{CYK}}\mybraces{n/2} + IO_{\mathcal{A}_{CYK}^*}\mybraces{n} + \mathcal{O}(1) & \text{otherwise}
        \end{cases}
    \end{equation*}
    Expanding the recurrence, by (\ref{eq:star-ub}), and adding the cost of processing the input we have:
    \begin{equation*}
        IO_{\mathcal{A}_{CYK}}\mybraces{n} \le \mathcal{O}(\frac{n^3\Gamma}{B\sqrt{M}} + \frac{n^2\Gamma}{B}\log{M} + \frac{n^2|\mathcal{R}_B|}{B} + \frac{n|\mathcal{R}_T|}{B})    
    \end{equation*}

\end{proof}

\subsection{On the tightness of the bounds}
When $M < c_1n$, where $c_1$ is a sufficiently small positive constant, the lower bound in Theorem~\ref{thm:cyk-lwb} simplifies to $\Omega\mybraces{\frac{n^3\Gamma}{B\sqrt{M}}}$, while the upper bound in Theorem~\ref{thm:cyk-ub} becomes $\mathcal{O}(\frac{n^3\Gamma}{B\sqrt{M}} + \frac{n^2|\mathcal{R}_B|}{B} + \frac{n|\mathcal{R}_T|}{B})$. Hence, if $(n|\mathcal{R}_B| + |\mathcal{R}_T|)/\Gamma \in \mathcal{O}(n^2/\sqrt{M})$, the lower-bound asymptotically matches the upper-bound. When disallowing recomputation, a similar analysis shows our upper and lower bounds match under the same condition for $M < c_2n$, where $c_2$ is a sufficiently small positive constant.


\section{Conclusion}
    This work has contributed to the characterization of the \io{}
complexity of Dynamic Programming algorithms by establishing asymptotically
tight lower bounds for a general class of DP algorithms sharing a common structure of sub-problem dependence. Our technique exploits common properties of the CDAGs corresponding to said algorithms, which makes it promising for the analysis of other families of
recursive algorithms, in particular other families of DP algorithms of interest. The generality of our technique is further showcased by the ability to extend it to more complex algorithms, such as the Cocke-Younger Kasami, for which we provide an (almost) asymptotically tight \io{} lower bound and a matching algorithm.

Our analysis yields lower bounds both for computations in which no intermediate value is computed more than once and for general computations allowing arbitrary recomputation. By doing so we reveal how when the size of the available cache is greater than $2n$ and $o\mybraces{n^2}$, schedules using recomputation can achieve an asymptotic reduction of the number of required \io{} operations by a factor $\Theta\left(n^2/\sqrt{M}\right)$. This is particularly significant as in many cases of interest (e.g. Matrix Multiplication, Fast Fourier Transform, Integer Multiplication) recomputation has shown to enable a reduction of the \io{} cost by at most a constant multiplicative factor. 

Although it is known that recomputation can decrease the I/O complexity of certain CDAGs, we are still far from characterizing those CDAGs for which recomputation proves effective. This overarching objective remains a challenge for any efforts aimed at developing a general theory of the communication requirements of computations.

\clearpage

\appendix
\section{Additional proofs}\label{app:additionalproofs}

\subsection{Proofs of technical statements in Section~\ref{sec:lwbnr}} \label{app:lwbnr}

\begin{lemma}\label{lem:mathtrick}
Given $n\in N^+$  non-negative real numbers $a_1,a_2,\ldots, a_n$ such that:
\begin{enumerate}
    \item $\sum_{i = 1}^{n} a_i^2 \ge T$, where $T \in  \mathbb R^+$, and 
    \item $\max_{i=1}^n a_i^2 \le L$, where $L \in \mathbb R^+$ and $L < T$.
\end{enumerate}
Then, $\sum_{i = 1}^{n}a_i \ge \frac{T}{\sqrt{L}}$
\end{lemma}

\begin{proof}
The lower bound comes from the intuitive strategy of making each $a_i$ as large as possible, i.e. if $q = \lfloor \frac{T}{L} \rfloor$, then assign $a_1, ..., a_q = \sqrt{L}$ and $a_{q+1} = \sqrt{T - qL}$. In this case, 
$\sum_{i = 1}^{q+1} a_i = q*\sqrt{L} + \sqrt{T  - qL} = q*\sqrt{L} + \sqrt{\frac{T}{L}  - q}*\sqrt{L} \ge q*\sqrt{L} + (\frac{T}{L}  - q)*\sqrt{L}  \text{   (as $\frac{T}{L} - q < 1$)} = \frac{T}{\sqrt{L}}$. To show the optimality of this strategy, we show that any other configuration is suboptimal.

Consider any other assignment $a_1, ..., a_n$ \footnote{that is not just a permutation of the same values}. We know that any such assignment must have at least two non-zero numbers $a_i, a_j < \sqrt{L}$. Let $k = a_i + a_j$. If $k \le \sqrt{L}$, since $k^2 > a^2 + b^2$, there exists some $c \in  \mathbb R^+$ such that $c^2 = a^2 + b^2$ and $c < a+b$, showing that our assignment $a_1, ..., a_n$ does not minimize $\sum_{i = 1}^{n}a_i$. 

Alternatively if $k > \sqrt{L}$, then $a_i^2 + a_j ^2 = a_i^2 + (k - a_i)^2 < \sqrt{L}^2 + (k - \sqrt{L})^2$ (the last step follows from the fact that $f(x) = x^2 + (k - x)^2$ defines an upwards-facing parabola with an axis of symmetry about $x = \frac{k}{2}$). This means that there exists some $c < k - \sqrt{L}$ such that $\sqrt{L}^2 + c^2 = a^2 + b^2$ and $L+c < a+b$, again showing that our assignment $a_1, ..., a_n$ does not minimize $\sum_{i = 1}^{n}a_i$.

\end{proof}

\subsection{\io{} complexity allowing recomputation with large memory available} \label{app:big-mem-recomp}
In this section, we analyse the \io{} complexity for CDAGS $G\in\mathcal{G}(n)$ allowing for recomputation with available memory $M \ge 2n$:
\begin{lemma}\label{lem:pebble-game}
For some $G \in \mathcal{G}(n)$, $n$ pebbles are necessary and sufficient to play the pebble game.
\end{lemma}

\begin{proof}
The necessity of $n$ pebbles to pebble $G$ can be shown using a proof identical to that of Lemma 10.2.2 in \cite{10.5555/522806}. To show that $n$ pebbles are also sufficient, we use a proof by strong induction. As a base case, when $G$ has only one input, clearly $1$ pebble is sufficient. When $G$ has two inputs, the claim also holds trivially.  

Now, assume $n \ge 3$ and that the claim holds for input sizes upto $n - 1$. Designate one pebble as an accumulator for $v_{1, n}$, leaving $n - 1$ free pebbles. We compute the leaves of $v_{1, n}$ sequentially, updating the accumulator in the process. Start with any leaf $l_k$ having $v_{1, 1+k}$ and $v_{k+2, n}$ as predecessors, which correspond to subproblems of size $k + 1$ and $n - k - 1$ repectively. We know that $\max{\mybraces{k + 1, n - k - 1}} \le n - 1$ and $\min{\mybraces{k + 1, n - k - 1}} \le n - 2$. Hence, there are sufficient pebbles to compute the larger subproblem, leave a pebble at the solution, and then compute the other subproblem. Then, compute $l_k$ and update the accumulator. Repeat the same process of the remaining leaves of $v_{1, n}$.
\end{proof}

\bigmemio*

\begin{proof}
This follows from Lemma~\ref{lem:pebble-game}. Keep one pebble each permanently at each of the input vertices, leaving $n$ free pebbles. From Lemma~\ref{lem:pebble-game}, this is sufficient to compute the CDAG without incurring any additional I/O costs (other than writing the output). 
\end{proof}

\subsection{Additional results on the Cooke-Younger-Kasami algorithm}\label{app:cyk-proofs}
\begin{algorithm}[H]
\caption{CYK algorithm} \label{alg:CYK}
\begin{algorithmic}[1]
\State \textbf{Input} $\{w_1, ..., w_n\}, \left(\mathcal{V}, \Sigma, \mathcal{R}, T\right)$
\State \textbf{Output} $\text{parse}$
\State $\text{parse} \leftarrow \text{False}$
\State $S \leftarrow \text{False}$
\If{$n = 0$}
    \If{$(T \rightarrow \epsilon) \in \mathcal{R}$}
        \State $\text{parse} \leftarrow \text{True}$
    \EndIf  
    \State \Return
\EndIf
\For{$i=1$ \textbf{to} $n$}
    \For{\textbf{each} unit production $(A \rightarrow a) \in \mathcal{R}$}
        \If{$a = w_i$}
            \State $S(i, i)[A] \leftarrow \text{True}$
        \EndIf
    \EndFor
\EndFor

\For{$l=2$ \textbf{to} $n$}
    \For{$i=1$  \textbf{to} $n-l+1$}
        \State $j\leftarrow i + l - 1$
        \For{$k=i$  \textbf{to} $j - 1$}
            \For{\textbf{each} binary production $(A \rightarrow BC) \in \mathcal{R}$}
                \If{$S(i, k)[B]$ \textbf{and} $S(k + 1, j)[C]$}
                    \State $S(i, j)[A] \leftarrow \text{True}$
                \EndIf
            \EndFor
        \EndFor
    \EndFor
\EndFor    
            
\If{$T \in S(1, n)$}
    \State{$\text{parse} \leftarrow \text{True}$}
\EndIf
\end{algorithmic}
\end{algorithm}

\cykcover*
\begin{proof}
    Consider any set of $x$ $\mathit{VR}$-vertices in $W_i$ containing $x_c$ vertices in column $i$ and $x - x_c$ vertices in row $i+1$. Since vertices in the same row (or column) don't interact, there are at most $x\mybraces{x - x_c} \le x^2/4$ interacting pairs. 

    Consider any two such distinct interacting pairs: $\mybraces{v^A_{k, i}, v^B_{i+1, l}}$ and $\mybraces{v^C_{m, i}, v^D_{i+1, n}}$ (where $v^A_{k, i}$ is a $\mathit{VR}$-vertex corresponding to variable $A$ in row $k$ and column $i$). Any leaf produced by $\mybraces{v^A_{k, i}, v^B_{i+1, l}}$ must correspond to a rule with right-hand side ``AB". From the assumption that every rule in $\mathcal{R}$ has a unique right-hand side, it follows that there is only one such rule and therefore $\mybraces{v^A_{k, i}, v^B_{i+1, l}}$ must produce a single leaf. 

    To show that the leaves produced by $\mybraces{v^A_{k, i}, v^B_{i+1, l}}$ and $\mybraces{v^C_{m, i}, v^D_{i+1, n}}$ belong to distinct $GR$-vertices, consider the following two cases:
    \begin{enumerate}
        \item If $k \ne m$ or $l \ne i$, it follows from Lemma \ref{lem:coverproperty} that the leaves produced belong to different roots, and therefore must also belong to different grammar roots.
        \item Otherwise, it must be the case that $A \ne C$ or $B \ne D$ (since the two interacting pairs are assumed to be distinct). It follows that the two leaves produced correspond to different rules in $\mathcal{R}$ and therefore belong to distinct grammar roots.
    \end{enumerate}

\end{proof}

\section{Valiant's DP Closure Algorithm} \label{app:valiant}
In this section, we introduce Valiant's DP Closure algorithm (\texttt{valiant\_closure}). For a more in-depth explanation, please refer to Cherng and Ladner \cite{cherng2005cache}. Consider any algorithm following Prototype Algorithm~\ref{alg:proto} with $n - 1$ inputs. In the section below, we represent \texttt{COMBINE} with ``$\cdot$'',  \texttt{AGGREGATE} with ``$+$'', and the \texttt{LEAST\_OPTIMAL\_VALUE} with ``$0$''.

To execute the algorithm, represent all subproblems in matrix $X \in A^{n, n}$ so that $S\mybraces{i, j}$ is represented at $X_{i, j+1}$. Initially, only the values $X_{i, i+1}$ corresponding to input values $S\mybraces{i, i}$ for $1 \le i \le n - 1$ are known, and hence the remaining entries in $X$ are set to 0. To fully compute all the subproblems (i.e. to compute the \textit{DP-closure}), \texttt{valiant\_closure} recursively divides X into $16$ equally-sized submatrices of dimension $\frac{n}{4}$\footnote{We assume that n is a power of $2$. This can always be made the case by padding $X$}. Label these matrices as follows:
 
\begin{equation*}
X = 
\begin{bmatrix}
X_{11} & X_{12} & X_{13} & X_{14} \\
       & X_{22} & X_{23} & X_{24} \\
       &        & X_{33} & X_{34} \\
       &        &        & X_{44}
\end{bmatrix}
\end{equation*}
We omit the lower-triangular submatrices as these will be $0$. During its execution, \texttt{valiant\_closure} makes recursive calls to itself as well as two subroutines: \textit{Valiant's Star Algorithm} (\texttt{valiant\_star}) and \textit{Matrix Multiply and Accumulate}. We describe each of these procedures below along with pseudocode.

\subsection{Matrix Multiply and Accumulate} \label{app:matmul}
Given three equally sized square matrices $A$, $B$, and $C$, this procedure computes $A + B \cdot C$. In the context of Valiant's DP Closure Algorithm, $A$ contains partially computed root vertices, and $B \cdot C$ represents the computation of leaves, where are then used to update $A$. A recursive in-place implementation of this function can be found in previous work (\cite{cherng2005cache, frigo1999cache}). 

\subsection{Valiant's Star Algorithm} \label{app:valiant-star}
Valiant's Star Algorithm takes in a matrix $X$ where the subproblems inthe top-left and bottom-right quadrants are already computed. The algorithm then recursively computes the remaining root-vertices. Pseudocode for this function can be found in Algorithm~\ref{alg:star}  

\begin{algorithm}[H]
\caption{Valiant's Star Algorithm (\texttt{valiant\_star})} \label{alg:star}
\begin{algorithmic}[1]
\State \textbf{Input} $X$ \Comment{Top-left and bottom-right are assumed to be fully computed}
\State \textbf{Output} $X$ \Comment{DP-Closure of X is computed in-place}
\If{$\text{dim}(X) = 2 $} \Comment{No Computation required}
    \State \textbf{return} 
\EndIf
\State $ \matC \leftarrow \texttt{valiant\_star}\mybraces{\matC}$
\State $X_{13} \leftarrow X_{13} + X_{12} \cdot X_{13}$
\State $ \matD \leftarrow \texttt{valiant\_star}\mybraces{\matD}$
\State $X_{24} \leftarrow X_{24} + X_{23} \cdot X_{34}$
\State $ \matE \leftarrow \texttt{valiant\_star}\mybraces{\matE}$
\State $X_{14} \leftarrow X_{14} + X_{12} \cdot X_{24}$
\State $X_{14} \leftarrow X_{14} + X_{13} \cdot X_{34}$
\State $ \matF \leftarrow \texttt{valiant\_star}\mybraces{\matF}$
\end{algorithmic}
\end{algorithm}

\subsection{Valiant's DP-Closure Algorithm}
Valiant's DP-Closure algorithm takes in matrix $X$ with only non-zero values corresponding to inputs. The root-vertices corresponding to the top-left and bottom-right submatrices are computed through recursive calls, exploiting the fact that these submatrices correspond to a smaller instantiation of the same problem. The remaining values in $X$ are then computed by calling \texttt{valiant\_star}. Pseudocode can be found in Algorithm~\ref{alg:closure}.

\begin{algorithm}[H]
\caption{Valiant's DP-Closure Algorithm (\texttt{valiant\_closure})} \label{alg:closure}
\begin{algorithmic}[1]
\State \textbf{Input} $X$
\State \textbf{Output} $X$ \Comment{DP-Closure of X is computed in-place}
\If{$\text{dim}(X) = 2 $} \Comment{Input is already computed}
    \State \textbf{return} 
\EndIf
\State $ \matA \leftarrow \texttt{valiant\_closure}\mybraces{\matA}$
\State $ \matB \leftarrow \texttt{valiant\_closure}\mybraces{\matB}$
\State $ X \leftarrow \texttt{valiant\_star}\mybraces{X}$

\end{algorithmic}
\end{algorithm}

\section{Detailed \io{} analysis for selected algorithms} \label{app:applications}

In this section, we apply the general results presented in Section~\ref{sec:main} on the \io{} complexity of DP algorithms following the structure outlined in Prototype Algorithm  $\mathcal{A}^*$ to obtain asymptotically tight bounds on the \io{} complexity of selected DP algorithms. 

\subsection{Matrix chain multiplication}\label{app:mcm}
\begin{algorithm}
\caption{Matrix Chain Multiplication}\label{alg:MCM}
\begin{algorithmic}[1]
\State \textbf{Input} $\mybraces{d_0, d_1, \dots, d_n}$
\State \textbf{Output} $S(1,n)$
\For{$i=1$ \textbf{to} $n$}
\State $S(i,i)\leftarrow 0$
\EndFor
\For{$l=2$ \textbf{to} $n$}
    \For{$i=1$  \textbf{to} $n-l+1$}
        \State $j\leftarrow i + l - 1$
        \State $S(i,j)\leftarrow \infty$
        \For{$k=i$  \textbf{to} $j - 1$}
            \State $q\leftarrow S(i, k) + S(k+1, j) + d_{i - 1}d_kd_j$
            \If{$q < S(i, j)$}
                \State $S(i,j) \leftarrow q$
            \EndIf
        \EndFor
    \EndFor
\EndFor
\end{algorithmic}
\end{algorithm}
\emph{Matrix chain multiplication} (MCM) is an optimization problem concerning the most efficient way to multiply a given sequence of $n$ matrices $A^1_{d_0\times d_1}\times A^2_{d_1\times d_2}\times \ldots \times A^n_{d_{n - 1}\times d_n}$, where $d_0, \ldots, d_n$ represent the dimensions of the matrices~\cite{clrs}. Pseudocode for the classical implementation of MCM ($\mathcal{A}_{MCM}$) is provided in Algorithm~\ref{alg:MCM}. While computing subproblem $S\mybraces{i,j}$, L-values correspond to the computations $S\mybraces{i,k}+S\mybraces{k+1,j}+d_{i - 1}d_kd_j$  for $1\leq i\leq j\leq n$ and $i\leq k < j$, from which the minimum value is selected. 

As suggested in \cite{cherng2005cache}, to map this onto the structure of Prototype Algorithm~\ref{alg:proto}, we can define the R-vertices to be thriples $\mybraces{a, b, c}$, where $a$ and $b$ are positive integers and $c$ is nonnegative. The \texttt{LEAST\_OPTIMAL\_VALUE} is $\infty$. The \texttt{COMBINE} ($\cdot$) operation is defined by $\mybraces{a, b, c} \cdot \mybraces{b, d, c'} = \mybraces{a, d, abd + c + c'}$, with all other cases evaluating to $\infty$. The \texttt{AGGREGATE} ($+$) operation is defined by $\mybraces{a, b, c} + \mybraces{a, b, c'} = \mybraces{a, b, \min \mybraces{c, c'}}$ and $\mybraces{a, b, c} + \infty = \mybraces{a, b, c}$, with all other cases evaluating to $\infty$. The $i$th input value is $\mybraces{d_{i - 1}, d_i, 0}$.

With this setup, MCM matches the structure of Prototype Algorithm~\ref{alg:proto}, with the caveat that the \texttt{COMBINE} and \texttt{AGGREGATE} operations are no longer unitary. However, this does not asymptotically affect the \io{} complexity as both these operations involve a constant number of unitary operations and the representation of $R$-vertices also takes a constant number of memory words. The following theorem follows:

\begin{theorem}\label{thm:iomcm}
    The number of \io{} operations executed by $\mathcal{A}_{MCM}$ given a sequence of $n$ matrices using a machine equipped with a cache memory of size $M < cn^2$, where $c$ is a sufficiently small positive constant, is bounded as: 
    \begin{equation*}
        IO_{G_{MCM}}\mybraces{n,m,B} \in \Theta\mybraces{\frac{n^3}{B\sqrt{M}}}.
        \end{equation*}
\end{theorem}

Other classic DP problems for problems such as \textit{boolean parenthesization} and \textit{finding the maximum/minimum value of an expression with arithmetic operations} have nearly identical structures and therefore lend themselves to the same analysis.

\subsection{Optimal polygon triangulation}
\begin{algorithm}
\caption{Optimal Polygon Triangulation} \label{alg:OPT}
\begin{algorithmic}[1]
\State \textbf{Input} $\{v_0, v_1, ..., v_n\}, w$
\State \textbf{Output} $S(1,n)$
\For{$i=1$ \textbf{to} $n$}
\State $S(i,i)\leftarrow 0$
\EndFor
\For{$l=2$ \textbf{to} $n$}
    \For{$i=1$  \textbf{to} $n-l+1$}
        \State $j\leftarrow i + l - 1$
        \State $S(i,j)\leftarrow \infty$
        \For{$k=i$  \textbf{to} $j - 1$}
            \State $q\leftarrow S(i, k) + S(k+1, j) + w(v_{i - 1}v_kv_j)$
            \If{$q < S(i, j)$}
                \State $S(i,j) \leftarrow q$
            \EndIf
        \EndFor
    \EndFor
\EndFor
\end{algorithmic}
\end{algorithm}
Given a convex polygon, the \emph{optimal polygon triangulation} (OPT) problem aims to find the cost of its \textit{optimal triangulation}. Specifically, one is given a convex polygon with vertices, $P = \{v_0, v_1, ..., v_n\}$ and a cost function $w: P \times P\times P \rightarrow \mathbb R$ specifying the cost of a single triangle (a common choice is perimeter). The goal is to divide $P$ into $n - 1$ disjoint triangles by drawing $n - 2$ chords between non-adjacent vertices (chords may not intersect except at their endpoints) while minimizing the total cost over the resulting triangles \cite{clrs}.

Pseudocode for OPT can be found in Algorithm~\ref{alg:OPT}. Subproblem $S\mybraces{i,j}$ represents the cost of the optimal triangulation of the polygon defined by vertices $\{v_{i - 1}, v_i, ..., v_j\}$. When computing subproblem $S\mybraces{i,j}$, L-vertices correspond to computations $S\mybraces{i,k}+S\mybraces{k+1,j}+w\mybraces{v_{i - 1}, v_k, v_j}$ for $1\leq i\leq j\leq n$ and $i\leq k < j$, from which the minimum value is selected.

To map OPT onto the Prototype Algorithm, one can use the same setup as Appendix~\ref{app:mcm}, modifying the \texttt{COMBINE} operation so that $\mybraces{a, b, c} \cdot \mybraces{b, d, c'} = \mybraces{w\mybraces{a, b, d} + c + c'}$. It follows that for $M < cn^2$, where $c$ is a sufficiently small positive constant:

\begin{equation*}
        IO_{G_{OPT}}\mybraces{n,m,B} \in \Theta\mybraces{\frac{n^3}{B\sqrt{M}}}.
\end{equation*}

\subsection{Construction of optimal binary search trees}
\begin{algorithm}
\caption{Construction of Optimal Binary Search Trees} \label{alg:BST}
\begin{algorithmic}[1]
\State \textbf{Input} $\{p_1, ..., p_{n - 1}\}, \{q_1, ..., q_{n}\}$
\State \textbf{Output} $S(1,n)$
\For{$i=1$ \textbf{to} $n$}
\State $S(i,i)\leftarrow q_i$
\State $C(i,i)\leftarrow q_i$
\EndFor
\For{$l=2$ \textbf{to} $n$}
    \For{$i=1$  \textbf{to} $n-l+1$}
        \State $j\leftarrow i + l - 1$
        \State $S(i,j)\leftarrow \infty$
        \State $C(i,j)\leftarrow C(i, j - 1) + p_{j - 1} + q_j$
        \For{$k=i$  \textbf{to} $j - 1$}
            \State $q\leftarrow S(i, k) + S(k+1, j) + C(i, j)$
            \If{$q < S(i, j)$}
                \State $S(i,j) \leftarrow q$
            \EndIf
        \EndFor
    \EndFor
\EndFor
\end{algorithmic}
\end{algorithm}
Given a sequence of $n - 1$ distinct keys $k_1, ..., k_{n - 1}$ in ascending order, the \textit{optimal binary search tree} (BST) problem involves building a binary search tree that minimizes the expected depth of search. To this end, one is also provided with probabilities $p_1, ..., p_{n - 1}$ and $q_1, ..., q_n$ where $p_i$ denotes the probability that key $k_i$ is searched for, and $q_i$ denotes the probability of a search for a non-existent key between $k_{i - 1}$ and $k_i$ \footnote{$q_1$ represents the probability of a search for a key less than $p_1$ and $q_n$ represents searches for keys more than $p_{n - 1}$}. Note that this implies that $\sum_{i = 1}^{n - 1} p_i + \sum_{i = 1}^{n}q_i = 1$. Pseudocode for BST can be found in Algorithm~\ref{alg:BST}

When constructing optimal binary search trees using dynamic programming, subproblem $S\mybraces{i,j}$  for $1\leq i\leq j\leq n$ represents the optimal cost of a binary search tree containing the keys $k_i, ..., k_{j - 1}$ \footnote{$S\mybraces{i,i}$ corresponds to having a tree with no keys, and has cost $q_i$}. When computing subproblem $S\mybraces{i,j}$, the L-vertices correspond to computations $S\mybraces{i,k}+S\mybraces{k+1,j}+ \sum_{l = i}^{j - 1} p_i + \sum_{l = i}^{j} q_i$ for $i\leq k < j$, from which the minimum value is selected.

To map BST onto Prototype Algorithm~\ref{alg:proto}, start by computing and writing $C\mybraces{i, j} = \sum_{l = i}^{j - 1} p_i + \sum_{l = i}^{j} q_i$ for $1 \le i \le j \le n$. To do so efficiently, one can utilize the following recursion: $C\mybraces{i, j} = C\mybraces{i, j - 1} + p_{j - 1} + q_j$ for $1 \le i < j \le n$. This allows each entry $C\mybraces{i, j}$ to be computed with a constant number of \io{} operations, and therefore completing this step incurs $\mathcal{O}(n^2)$ \io{} operations.

Now, define the R-vertices to be thriples $\mybraces{i, j, k}$, where $i$ and $j$ are positive integers and $k$ is nonnegative. The \texttt{LEAST\_OPTIMAL\_VALUE} is $\infty$. The \texttt{COMBINE} ($\cdot$) operation is defined by $\mybraces{i, j, k} \cdot \mybraces{j + 1, l, k'} = \mybraces{i, l, k + k' + C\mybraces{i, l}}$, with all other cases evaluating to $\infty$. The \texttt{AGGREGATE} ($+$) operation is defined by $\mybraces{i, j, k} + \mybraces{i, j, k'} = \mybraces{i, j, \min \mybraces{k, k'}}$, $\mybraces{i, j, k} + \infty = \mybraces{i, j, k}$, with all other cases evaluating to $\infty$. The $i$th input value is $\mybraces{i, i, C\mybraces{i, i}}$. With these modification, we conclude that:

\begin{theorem}
    The number of \io{} operations executed by $\mathcal{A}_{BST}$ when computing an optimal binary search tree given $n - 1$ unique keys on a machine equipped with a cache memory of size $M < cn^2$, where $c$ is a sufficiently small positive constant, is bounded as:
    \begin{equation*}
        IO_{G_{BST}}\mybraces{n,m,B} \in \Theta\mybraces{\frac{n^3}{B\sqrt{M}}}.
    \end{equation*}
\end{theorem}

\end{document}